\documentclass[acmsmall, nonacm]{acmart}

\usepackage[ruled]{algorithm2e}
\usepackage{algpseudocode}
\usepackage[capitalize]{cleveref}
\usepackage{subcaption}
\usepackage{graphicx}
\usepackage{listings}
\usepackage{mathpartir}
\usepackage{tikz}
\usepackage{pgfplots}
\usepackage{color}

\pgfplotsset{width=0.8\textwidth,compat=1.9}
\usepgfplotslibrary{fillbetween}
\usepgfplotslibrary{groupplots}
\usepgfplotslibrary{statistics}

\pgfplotsset{
    discard if not/.style 2 args={
        x filter/.code={
            \edef\tempa{\thisrow{#1}}
            \edef\tempb{#2}
            \ifx\tempa\tempb
            \else
                
            \fi
        }
    }
}

\usetikzlibrary{arrows}
\usetikzlibrary{positioning}

\tikzset{
farrow/.tip={stealth}
}

\crefname{lstlisting}{listing}{listings}
\Crefname{lstlisting}{Listing}{Listings}

\graphicspath{ {./figures/} }

\keywords{Multiverse Debugging, Embedded devices, WebAssembly}

\newtheorem{theorem}{Theorem}

\usepackage{tikz}
\usepackage{xspace}
\usepackage{ifthen}

\setcopyright{rightsretained}
\acmDOI{10.1145/3763136}
\acmYear{2025}
\acmJournal{PACMPL}
\acmVolume{9}
\acmNumber{OOPSLA2}
\acmArticle{358}
\acmMonth{10}

\begin{document}

\title{MIO: Multiverse Debugging in the Face of Input/Output}
\subtitle{Extended Version with Additional Appendices}

\author{Tom Lauwaerts}
\orcid{0000-0003-1262-8893}
\affiliation{%
  \institution{Universiteit Gent}
  \city{Ghent}
  \country{Belgium}
}
\email{Tom.Lauwaerts@UGent.be}

\author{Maarten Steevens}
\orcid{0009-0002-6339-0467}
\affiliation{%
  \institution{Universiteit Gent}
  \city{Ghent}
  \country{Belgium}
}
\email{Maarten.Steevens@UGent.be}

\author{Christophe Scholliers}
\orcid{0000-0002-2837-4763}
\affiliation{%
  \institution{Universiteit Gent}
  \city{Ghent}
  \country{Belgium}
}
\email{Christophe.Scholliers@UGent.be}

\hyphenation{mi-cro-con-troller}
\hyphenation{mi-cro-con-trollers}

\definecolor{color5}{RGB}{105,214,255}
\definecolor{color4}{RGB}{37,187,255}
\definecolor{color3}{RGB}{0,153,248}
\definecolor{color2}{RGB}{0,100,210}
\definecolor{color1}{RGB}{0,30,130}

\definecolor{color0}{RGB}{242,48,90}

\begin{abstract}
Debugging non-deterministic programs on microcontrollers is notoriously challenging, especially when bugs manifest in unpredictable, input-dependent execution paths.
A recent approach, called multiverse debugging, makes it easier to debug non-deterministic programs by allowing programmers to explore all potential execution paths.
Current multiverse debuggers enable both forward and backward traversal of program paths, and some facilitate jumping to any previously visited states, potentially branching into alternative execution paths within the state space.

Unfortunately, debugging programs that involve input/output operations using existing multiverse debuggers can reveal inaccessible program states, i.e. states which are not encountered during regular execution.
This can significantly hinder the debugging process, as the programmer may spend substantial time exploring and examining inaccessible program states, or worse, may mistakenly assume a bug is present in the code, when in fact, the issue is caused by the debugger.

This paper presents a novel approach to multiverse debugging, which can accommodate a broad spectrum of input/output operations.
We provide the semantics of our approach and prove the correctness of our debugger, ensuring that despite having support for a wide range of input/output operations the debugger will only explore those program states which can be reached during regular execution.

We have developed a prototype, called MIO, leveraging the WARDuino WebAssembly virtual machine to demonstrate the feasibility and efficiency of our techniques.
As a demonstration of the approach we highlight a color dial built with a Lego Mindstorms motor, and color sensor, providing a tangible example of how our approach enables multiverse debugging for programs running on an STM32 microcontroller.
\end{abstract}

\begin{CCSXML}
<ccs2012>
   <concept>
       <concept_id>10010520.10010553.10010562.10010564</concept_id>
       <concept_desc>Computer systems organization~Embedded software</concept_desc>
       <concept_significance>300</concept_significance>
       </concept>
   <concept>
       <concept_id>10011007.10011006.10011039.10011311</concept_id>
       <concept_desc>Software and its engineering~Semantics</concept_desc>
       <concept_significance>500</concept_significance>
       </concept>
   <concept>
       <concept_id>10011007.10011074.10011099.10011102.10011103</concept_id>
       <concept_desc>Software and its engineering~Software testing and debugging</concept_desc>
       <concept_significance>500</concept_significance>
       </concept>
   <concept>
       <concept_id>10011007.10011006.10011066.10011069</concept_id>
       <concept_desc>Software and its engineering~Integrated and visual development environments</concept_desc>
       <concept_significance>300</concept_significance>
       </concept>
 </ccs2012>
\end{CCSXML}

\ccsdesc[300]{Computer systems organization~Embedded software}
\ccsdesc[500]{Software and its engineering~Semantics}
\ccsdesc[500]{Software and its engineering~Software testing and debugging}
\ccsdesc[300]{Software and its engineering~Integrated and visual development environments}

\maketitle

\section{Introduction}

Debugging non-deterministic programs is a challenging task, since bugs may only appear in very specific execution paths \cite{mcdowell89, gurdeep22}.
This is especially true for microcontroller programs, which typically interact heavily with the environment.
This makes reproducing bugs unreliable and time-consuming, a problem that traditional debuggers do not account for.
By contrast, multiverse debugging \cite{torres19} is a novel technique that solves this problem by allowing programmers to explore all possible execution paths.
A multiverse debugger allows users to move from one execution path to another, even jumping to arbitrary program states in parallel execution paths. This entails traveling both forwards and backwards in time, i.e. a multiverse debugger is also a time travel debugger.
So far, existing implementations work on abstract execution models to explore all execution paths of a program \cite{torres19,pasquier22,pasquier23,pasquier23a}. Within these semantics, only the internal state of the program is controlled.

Unfortunately, debugging programs that involve input/output (I/O) operations using existing multiverse debuggers can reveal inaccessible program states that are not encountered during regular execution.
This is known as the probe effect \cite{gait86}, and can occur in multiverse debuggers when they do not account for the effect of I/O operations on the external environment when changing the program state.
Encountering such states during the debugging session can significantly hinder the debugging process, as the programmer may mistakenly assume a bug is present in the code, when in fact, the issue is caused by the debugger.
In this paper, we investigate how we can scale multiverse debugging to programs running on a microcontroller which interacts with the environment through I/O operations.
This introduces three new challenges.

First, the effect of output operations on the environment can influence later states in the execution path, for example when a robot drives forward.
Therefore, when stepping backwards, the changes made to the environment by the program must be reverted to stay consistent with normal execution.
Otherwise, the debugger may enter inaccessible execution paths.
This far-reaching probe effect, is especially difficult to control when making arbitrary jumps in the execution tree.
Second, input operations from the external environment make it difficult to maintain reproducibility \cite{frattini16}.
There are often too many possible execution paths to explore due to an infinite range of inputs.
Furthermore, it is impractical for developers to enumerate all possible inputs, or to perfectly configure the environment to achieve a specific execution path every time.
Without a way to handle the large ranges of possible inputs, the reproducibility of non-deterministic bugs in the multiverse debugger becomes challenging.
Third, due to the hardware limitations of the microcontrollers that we target it is unfeasible to run the multiverse debugger entirely on the microcontroller.
We thus need to expand multiverse debugging so that it can be used even in such a restricted environment. 

\subsection{Overview}

In this paper, we present a new approach to multiverse debugging tackling the three main challenges posed in the introduction making several key scientific contributions.

First, we establish a well-defined set of I/O operations that are \emph{deterministically reversible}, which is a crucial requirement for the correctness proof of our multiverse debugger.
Second, we formalize our debugger as a \emph{remote debugger} with \emph{sparse snapshotting} semantics to accommodate the hardware limitations of microcontrollers.
Third, we define our semantics such that the multiverse debugger can instrument input operations, which is needed to allow the programmer to interactively explore the execution tree.
Finally, we formulate our semantics in a \emph{general and abstract} way with \emph{novel correctness criteria}, providing valuable guidelines to formalize debuggers in general. 

To demonstrate the practical applicability of our formal semantics we present MIO, a prototype multiverse debugger build on top of the WARDuino WebAssembly virtual machine~\cite{lauwaerts24a}. 
MIO can debug code on various microcontrollers, including the ESP32, STM32, and Raspberry Pi Pico. 

Our implementation also makes several technical contributions. 
It is, to our knowledge, the first \emph{concrete multiverse debugger} to operate on a real-world, multi-language ecosystem (i.e., WebAssembly), moving beyond existing multiverse debuggers which have primarily focused on \emph{abstract models}~\cite{torres19,pasquier22,pasquier23,pasquier23a}
Furthermore, MIO is a \emph{remote} multiverse debugger specifically tailored for resource-constrained microcontrollers. 
To address the significant performance cost of state capture on these devices, this work implements and benchmarks for the first time a \emph{sparse snapshotting} approach. 
Our results confirm that this technique is essential for achieving practical performance, mitigating the high overhead of traditional (full) snapshotting.
 
We start the article with an overview of multiverse debugging with I/O in \cref{sec:practice}, discussing the different parts of our prototype and the underlying challenges they overcome.
This is followed by \cref{sec:multiverse-debugger} with a detailed discussion of our formalism.
We start this section with a discussion of the requirements for the I/O operations, and a brief summary of the key takeaways of the formalism.
Next, in \cref{sec:implementation} we discuss our prototype implementation and show how the \emph{deterministically reversible} primitives are implemented.
In \cref{sec:evaluation} we provide a quantitative evaluation of our snapshotting approach and discuss a case study showing how the technique works in practice.
We conclude the article with the related work and conclusion in sections \ref{sec:related} and \ref{sec:conclusion}.    

\section{MIO: Multiverse Input/Output Debugger in Practice}\label{sec:practice}

Before we discuss the details of our contributions and the implementation, we give an overview of how our multiverse debugger, MIO, works in practice. We will use a simple example to illustrate the different concepts of our multiverse debugger.

\subsection{Example: A Light Sensor Program}

\begin{figure}
    \centering
    \includegraphics[width=0.8\textwidth]{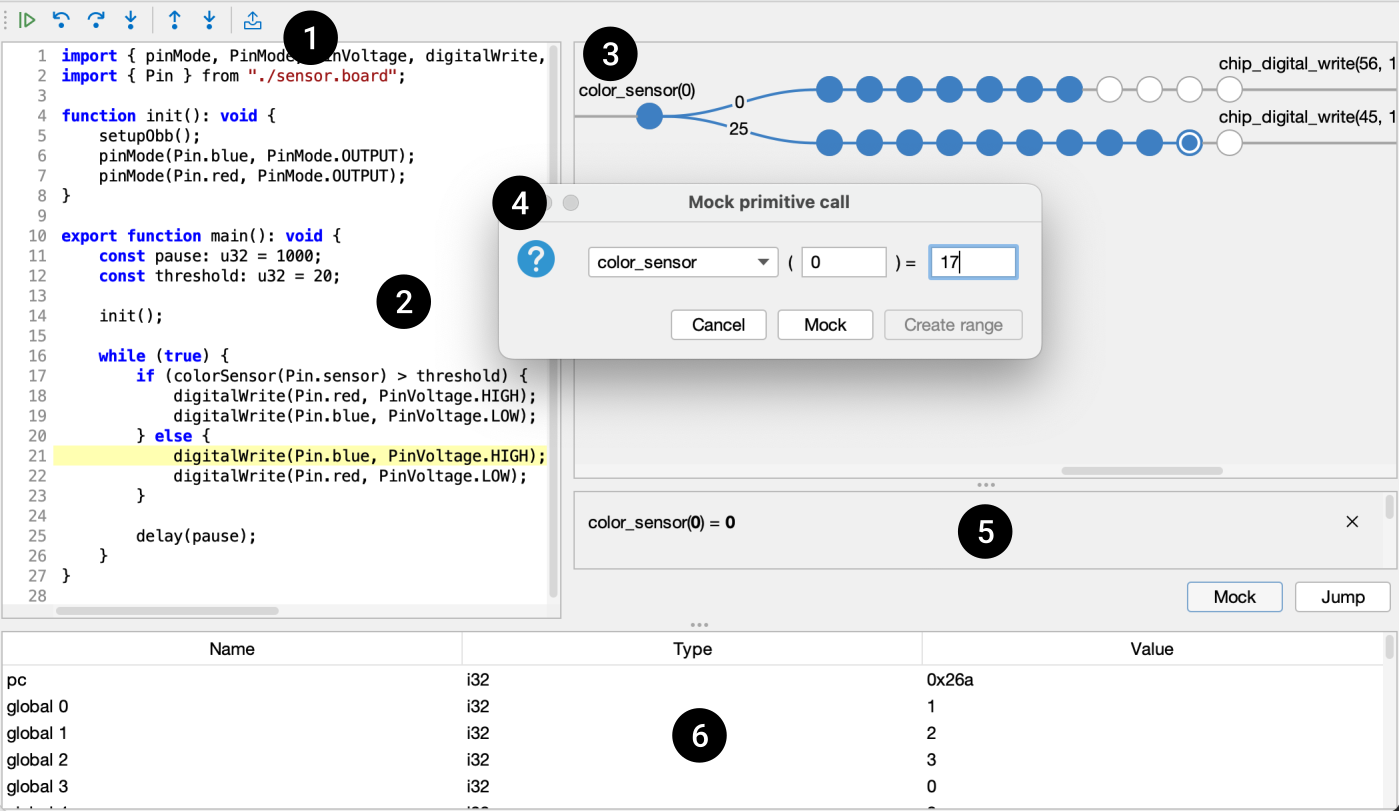}
    \caption{Screenshot of the MIO debugger debugging a small light sensor program.
        Top (1): the debug operations; pause or continue, step back, step over, step into, step to the previous line, step to the next line, and update software.
        Left pane (2): source code.
        Top-right pane (3): the multiverse tree. %, where each node is a WebAssembly instruction.
        Popup window (4): window to add mocked input.
        Bottom-right pane (5): the current mocked input.
        Bottom pane (6): general debugging information, such as the global and local variables, and the current program counter.
    }
    \label{fig:multiverse-debugger}
\end{figure}

Consider a simple program that reads a value from a color sensor, to measure the light intensity, and turns on an LED with a color corresponding to the value read.
The left side of \cref{fig:multiverse-debugger} (2) shows the example program written in AssemblyScript \cite{battagline21}.
Reading the value from the color sensor happens through the \emph{colorSensor} function, and changing the color of the LED happens through the \emph{digitalWrite} function.
In an infinite loop, the program reads a value from the color sensor.
If the value is below a threshold, the red LED is turned on, and otherwise the blue LED is turned on.
At the end of the loop the program waits for 1000 milliseconds before starting the next iteration.

\subsection{The Frontend of the Multiverse Debugger}

\Cref{fig:multiverse-debugger} shows the MIO multiverse debugger in action.
While debugging, the program and the debugger backend run on a microcontroller connected to a computer running the debugger frontend, shown in the screenshot.
The left pane (2) shows the program being debugged, in this case the light sensor program.
In the top left corner (1), the user can see the debug operations available, which are in order; pause or continue, step back, step over, step into, step to the previous line, step to the next line, and finally, update the code on the microcontroller.
The top-right pane (3) shows the multiverse tree, where each edge represents a WebAssembly instruction.
In MIO all nodes are considered unique states, which means that loops are unrolled and no states will ever collapse into each other, making this graph always a rooted tree.
MIO allows input values to be mocked using the \emph{Mock} button, which triggers a popup window (4) where users can specify the mocked return value for a given input primitive.
The \emph{Create range} button can be used to add a range of new branches to the multiverse tree.
The mocked primitives show up in the bottom-right pane (5), where they can also be removed again.
The bottom pane (6) shows the global and local WebAssembly variables, and other program state information.

While debugging the light sensor program, the user can pause the program at any moment.
At this point, the right pane will show the already explored paths of the multiverse tree.
It is important to note that the multiverse tree does not correspond to a control flow graph, but shows every succeeding program state.
Every edge in the tree represents the concrete execution of a single WebAssembly instruction, and every node represents the program state after that instruction.
The multiverse tree in \cref{fig:multiverse-debugger} labels every node before the execution of a primitive with the primitive name, and labels the outgoing edges with the associated return value.

\subsection{Debugging with the Multiverse Debugger}

A developer can use the multiverse debugger to explore the execution of the light sensor program, using the debug operations in the top left corner, following normal debugging conventions.
While the right panel continually updates to show where in the program's execution the user is.
However, the multiverse debugger allows for more than just stepping forwards and backwards through the program.
The user can also jump to any node in the multiverse tree by simply clicking on that node, and explore the program from that point onwards with the available debugging instructions.

\subsection{Exploring the multiverse tree}

The screenshot in \cref{fig:multiverse-debugger} shows the program paused at line 21, after reading a value from the color sensor.
This value was 0, and if the user steps forward, the program will turn on the blue LED (pin 56) as shown by the multiverse tree.
In an earlier stage, the user actually read a value of 25.
Using the debugger, it is possible to explore that execution path again, without having to recreate the situation where the sensor reads 25.
This is done by clicking on any node on the execution path where the sensor reads 25, which is shown as a blue node with a white circle in it.
If the user presses the \emph{Jump} button, the debugger will traverse the blue-indicated path in the multiverse tree, from the current node to the selected node.
When reading the sensor value, the debugger will use mock input actions, to override the normal behavior of the sensor and return the value 25.

The mocking mechanism is a crucial part of the multiverse debugger, since it allows the user to explore different execution paths without having to recreate the exact conditions that led to that path.
It is therefore not only a part of the automatic traversal of the multiverse tree, but also a part of the manual exploration of the tree.
While a user is stepping forwards through a program, and a non-deterministic input primitive is encountered, they can press the \emph{Mock} button below the right pane instead of the \emph{step} button.
When doing so, window (4) shown in \Cref{fig:multiverse-debugger} is used. This window allows the user to specify that the virtual machine should return a specific value when a primitive is called with certain arguments.
This allows users to easily reproduce bugs that might not occur with the sensor values read from the current environment.

Now that we have established how our multiverse debugger works in practice, and its terminology, we will discuss the challenges in creating a multiverse debugger for microcontrollers.

\subsection{Challenge C1: Inconsistent External State during Backwards Exploration}
When moving backwards in time, not only the internal, but also the external state has to be restored.
A multiverse debugger that allows for backwards exploration, but does not handle the external state, will be able to create impossible situations.
For instance, if the light sensor program is paused when the red LED has just been turned on, the debugger can move backwards in time to when the light intensity is measured.
Reading a different value as you step forwards again, this time, the blue LED could be turned on, even though the red LED is still on.
This is an impossible situation, since the LEDs can never be on at the same time in a normal execution of the program.
Stepping back over line 21, therefore, requires a compensating action, which turns off the red LED.
Luckily, such actions can be defined for many different output operations~\cite{schultz20,laursen18}.

\subsection{Challenge C2: Exploring Non-deterministic Input in Multiverse Debuggers}

When a program's execution path is determined by input from the environment, the multiverse debugger needs to explore different execution paths by traveling back in time and changing the input to the program.
However, it is impractical for developers to cover all possible inputs, or to configure the environment exactly right for a specific execution path every time.
To solve this problem, we propose a new approach to multiverse debugging that uses time travel debugging in combination with virtual machine instrumentation to mock input values.
It is important that no impossible values are used for mocking, since they can lead to inaccessible program states.

\subsection{Challenge C3: Keeping Track of the Program State Efficiently}

In order for multiverse debugging to work on a microcontroller it has to keep track of the program state and its output effects on the external environment.
However, tracking this information on microcontrollers is not feasible due to the limited memory capacity of these devices.
Therefore, the MIO debugger is a remote debugger, which enables minimal interference with the microcontroller, and allows information to be stored on a more powerful computer.
Even so, it is not feasible to take snapshots for every executed instruction as it would slow down the program significantly, and the size of the snapshot history would quickly become unmanageable.

In time travel debugging, this problem is usually solved by taking snapshots at regular intervals, called checkpoints.
When moving backwards in time, the debugger can then jump to the nearest checkpoint, and replay the program's execution from that point.
To make a checkpointing system work correctly for multiverse debugging, we had to adopt a slightly different approach.
Instead of only taking snapshots at regular intervals, we also take snapshots after each input or output action.
By carefully choosing when snapshots are taken we can ensure the correctness of the replay mechanism. Using this approach, we can significantly reduce the run-time overhead, making multiverse debugging practically usable on microcontrollers.

\section{A Multiverse Debugger for WebAssembly}\label{sec:multiverse-debugger}
In this section, we discuss the operation of our multiverse debugger through a small-step semantic defined over a stack-based language.
We use WebAssembly as an example language, since it has a full and rigorous language semantics~\cite{haas17,rossberg19,rossberg23}.
However, our small-step rules include very few details specific to WebAssembly.
Therefore, we believe that the principles of our multiverse debugger can be applied to any stack-based language, with minimal effort.

\subsection{Overview of the Formal Semantics}
This section provides a high-level overview of our formalization, abstracting away some of the (important) details. 
We present this overview in the hope that it makes it easier to apply our techniques to other programming languages. 
We follow the approach of \citet{torres19} in presenting a debugger formalisation recipe.
The formalization of our multiverse debugger for an arbitrary language, consists of the steps shown in~\cref{fig:recipe}. 
 For each step, we also indicate the relevant section demonstrating its specific application to WebAssembly.

\subsubsection*{Base Language Semantics}
First, the small-step operational semantics (SOS) for the base language are defined. 
This establishes a formal transition system that precisely describes how a program's configuration evolves one computational step at a time. 
The semantics are defined as a transition relation, relating the program state before taking a step $c$ to the program state after taking a step $c'$. 
We will use the following notation, $c \hookrightarrow c'$, where $c$ and $c'$ are program configurations (e.g., the state of the entire program including the stack, memory, etc.. ).  
This formal model serves as the base layer upon which all subsequent semantics are built. 

\begin{figure}
\noindent\rule[0.5ex]{\linewidth}{0.5pt}
\begin{enumerate}
    \item \textbf{Base Language Semantics:} First, define the small-step operational semantics for the base language (\Cref{sec:base_semantics}). 
    This provides the base model of program execution.

    \item \textbf{Well-Behaved I/O:} Extend the base language semantics to ensure all I/O primitives (\Cref{sec:ioprimitives}) conform to the requirements for reversible I/O (\Cref{sec:iorequirements}). This step is crucial for preventing the exploration of inaccessible states.

    \item \textbf{Debugger Semantics:} Define the \emph{remote} multiverse debugger in terms of the extended base language semantics. 
    This involves the following steps:
    \begin{enumerate}
        \item Define the debugger's state, which encapsulates the program's state in the extended base language (\Cref{sec:configuration}).
        \item Define forward execution steps for the debugger, incorporating state snapshotting to enable exploration (\Cref{sec:forwards}).
        \item Define backward execution steps that allow reverting to previously snapshotted states with compensation (\Cref{sec:backwards}).
        \item Define I/O mocking operations that allow the debugger to control and observe interactions with the environment (\Cref{sec:mocking}).
    \end{enumerate}

    \item \textbf{Correctness Proof:} Finally, prove the correctness of the debugger (\Cref{sec:correctness}). 
    Specifically, show that the debugger only traverses program states that are reachable during a regular, non-debugging execution.
\end{enumerate}
\noindent\rule[0.5ex]{\linewidth}{0.5pt}
  \caption{Overview of steps to formalize a MIO multiverse debugger for a given base language.}
  \label{fig:recipe}
\end{figure}

\subsubsection*{Well-Behaved I/O}
Unconstraint I/O operations present a challenge for reversibility, to address this the base semantics are extended by introducing \textbf{compensating actions}. 
The semantics of each I/O primitive are redefined to return not just a value $v$, but a pair $(v, r)$, where $r$ is the specific ``undo'' instruction for that action.
For instance, consider a primitive \texttt{set\_pin\_high(13)} that produces a return value $v$ (i.e. a code indicating that the action completed correctly). 
In our semantics this action is adjusted to also generates its unique compensating action $r$. 
This compensating action $r$ is a function that encapsulates the logic for the reverse operation, for example \texttt{set\_pin\_low(13)}.

\subsubsection*{Debugger Semantics}
The semantics of the \emph{remote} debugger can be modeled as another transitioning relation $ d_c \hookrightarrow_d d_c'$~\footnote{Note the $d$ subscript on the arrow} working over a debugging configuration $d_c$. 
This debugging configuration $d_c$ consists of the debugger state (i.e. paused or running), a message box, the underlying program state $c$ and a list of compensating actions. 
Interestingly the debugger's transition relation, $d_c \hookrightarrow_d d_c'$, is defined in terms of the underlying program's transition relation, $c \hookrightarrow c'$, which it invokes to drive the execution.
For example, consider a \texttt{step} command issued to the debugger under the form of an incoming message. 
The debugger, currently in a \texttt{paused} state with program state $c$, will use the underlying base semantics to compute its next state, i.e. it takes one step in the base language semantics $c \hookrightarrow c'$. Once this underlying transition is determined, the debugger completes its own transition, $d_c \hookrightarrow_d d_c'$, by updating its configuration to contain the new program state $c'$. 

During a standard program execution of an I/O action, the forward semantics ($c \hookrightarrow c'$) use the return value $v$ and simply discard the compensating action $r$~\footnote{This would of course be optimized in an actual implementation.}. 
In contrast, in the debugger's semantics these compensating actions are captured. 
This allows for valid backward transitions, to reverse the \texttt{set\_pin\_high(13)} operation from the example, the debugger executes its saved compensating action, restoring the system to its previous state by setting the pin low.

While this overview already gives a high-level summary of the core semantics, it omits some important details. 
In the following sections we present the precise inference rules for forward execution with snapshotting (\Cref{sec:forwards}) and backward execution using compensation (\Cref{sec:backwards}) which make our approach scalable for use on microcontrollers.  
Finally, we will define the semantics for I/O mocking (\Cref{sec:mocking}).
We start with detailing the requirements for I/O operations.

\subsection{Requirements for I/O operations}
\label{sec:iorequirements}
While our multiverse debugger can deal with a large set of I/O operations, there are some limitations that make it impossible to support certain I/O operations.
Intuitively, our multiverse debugger supports non-deterministic input primitives as long as the range of possible input values is known.
Output primitives are supported as long as they are atomic, and are deterministically reversible, i.e. after reversing an output operation the environment will be in the same state as before applying the operation. In \cref{sec:Discussion}, we discuss the limitations and opportunities for future improvements that these requirements impose on our debugger.

\paragraph{Input primitives}
Input primitives are allowed to be non-deterministic as long as the \textit{range} of the input primitives is known, i.e. a temperature sensor might have a range between -20 degrees till 160 degrees. Knowing the range of our input primitives is important so that the debugger can be instrumented to only sample values that can actually be observed during normal execution.

\paragraph{Output primitives}
First, we require output primitives to be synchronous and atomic, i.e.  all side effects from the operation have to be fully completed during the call of the operation in the virtual machine. 
Second, for a given execution of the I/O operation, there must be a \textit{deterministic compensating action}. 
This compensating action reverses the effects of the forward execution, bringing the environment back to a state before executing the actions. 

\paragraph{Predictable dependencies}
We assume \textit{predictable} dependencies of the I/O operations are known, for example consider a setup where an LED is directly pointed towards a light sensor. During regular execution, turning the LED on will directly influence the possible values which can be read, i.e. there is a dependency between the output pin and the possible sensor values which can be read. 
Our MIO debugger, has initial support for expressing such simple dependencies, in the semantics however, we take abstraction and assume that sensor values are independent. 

\subsection{WebAssembly Language Semantics}
\label{sec:base_semantics}
WebAssembly is a low-level, portable binary code format designed for safe and efficient execution on various platforms.
Its formal semantics are grounded in a stack-based virtual machine, where instructions and values operate on a single stack with strict typing to guarantee fast static validation.
We base our formalisation on the semantics from the original WebAssembly paper by \citet{haas17}, where the core semantics include structured control flow (blocks, loops, and conditionals) and memory management via linear memory.
WebAssembly intentionally excludes external interface definitions, including I/O operations, to minimize its dependence on platform-specific details.
This design choice enables us to deliberately sculpt the I/O operations so that they are deterministically reversible, a necessary condition for multiverse debugging.

The execution of a WebAssembly program is defined by a small-step reduction relation, denoted as $\hookrightarrow_{i}$ where, $i$ refers to the index of the currently executing module. 
The relation $\hookrightarrow_{i}$ is defined over a configuration $\{s;v^*;e^*\}$, with global store $s$, local values $v^*$, and the current stack of instructions $e^*$.
The reduction rules take the form $\{s;v^*;e^*\} \hookrightarrow_{i} \{s';v'^*;e'^*\}$.
Important for our semantics, the global store $s$ contains instances of modules, tables, and memories.
The global store allows access to any function within a module instance, denoted as $s_{\textit{func}}(i,j)$, where $i$ represents the module index and $j$ corresponds to the function index.

\subsection{Extending WebAssembly with Primitive I/O Operations}\label{sec:webassembly}

Since multiverse debuggers explore all possible execution paths, they have to be able to reproduce the program's execution, even when non-determinism is involved.
It is therefore important to consider where non-determinism is introduced in the program, and how it can be handled.
Since the WebAssembly semantics on their own are fully deterministic, we can choose precisely where non-determinism is introduced in our system.
In the context of microcontrollers, non-deterministic input is unavoidable. % since the external environment is inherently non-deterministic.
Our system therefore limits non-determinism exclusively to the input\footnote{In this work we do not consider parallelism as a source of non-determinism, since this has been examined thoroughly by the original paper on multiverse debugging by \citet{torres19}.}.
This means that each branch in the execution tree can be traced to a different input value.
The output primitives in MIO, on the other hand, are deterministic in terms of the program state.
\Cref{sec:mocking} discuss how the debugger can reproduce non-deterministic input reliably through input mocking.

\subsubsection{Primitive I/O operations}
\label{sec:ioprimitives}

\begin{figure*}
        \[
	\begin{array}{ l l c l }
            \emph{(WebAssembly Program state)} & K                         & \Coloneqq  & \{ s, v^*, e^* \} \\
            \emph{(Global store)}              & s                         & \Coloneqq  & \{ \textsf{inst } \textit{inst}^*, \textsf{tab } \textit{tabinst}^*, \textsf{mem } \textit{meminst}^*, \colorbox{lightgray}{\textsf{prim} $P$} \} \\
            %\emph{(Instances)}                 & \textit{inst}             & \Coloneqq  & \{ \textsf{func } \textit{cl}^*, \textsf{glob } v^*, \textsf{tab } i^?, \textsf{mem } i^?, \colorbox{lightgray}{\textsf{prim} $P$} \} \\
            \emph{(Primitive table)}           & \colorbox{lightgray}{$P$} & \Coloneqq  & \colorbox{lightgray}{$p^*$} \\
            \\
            \hline \\
            \emph{(Primitive)}                 & \colorbox{lightgray}{p}   & \coloneq  & \colorbox{lightgray}{$f : v^* \rightarrow \{ \textsf{ret } v , \textsf{cps } r \}$} \\
            \emph{(Compensating action)}       & \colorbox{lightgray}{r}   & \coloneq  & \colorbox{lightgray}{$f : \epsilon \rightarrow \epsilon$} \\
        \end{array}
	\]
    \caption{The configuration for the reversible primitives embedded in the WebAssembly semantics from the original paper by \citet{haas17}, the differences are highlighted in gray. \emph{Top:} The WebAssembly semantics extended with a primitive table. \emph{Bottom:} The signatures of primitive and their compensating actions.}
	\label{fig:prim-def}
\end{figure*}

We extend WebAssembly with a set of primitives $P$ that fulfil the prerequisites outlined above.
\Cref{fig:prim-def} shows the definition of the primitives in WebAssembly.
Similar to functions in WebAssembly, each primitive can be identified by a unique index $j$.
Looking up primitives is done through the global $P(j)$, and calling a primitive returns both the return value $v$ and the compensating action $r$.
The function $r$ compensates, or reverses, the effects of the primitive, but takes no arguments and returns nothing as indicated by its type $\epsilon \rightarrow \epsilon$.
There is no need for any arguments since the compensating action is generated uniquely for each execution of the primitive.

\begin{figure}
	\begin{mathpar}
                \inferrule[(\textsc{input-prim})]
       	            { 
                        P(j) = p \\
                        p \in P^{In} \\
                        v \in \lfloor p(v_0^*)_{ret} \rfloor \\
                    }
                    {
                      \{ s ;v^*; v^*_0 (call \; j) \} 
                      \hookrightarrow_{i}
                      \{ s ;v^*; v \} }

                                    \inferrule[(\textsc{output-prim})]
       	            { 
                        P(j) = p \\
                        p \in P^{Out} \\
                        \lfloor p(v_0^*)_{ret} \rfloor = v \\
                    }
                    {
                      \{ s ;v^*; v^*_0 (call \; j) \} 
                      \hookrightarrow_{i}
                      \{ s ;v^*; v \} }
	\end{mathpar}
        \caption{Extension of the WebAssembly language with non-deterministic input primitives.}
	\label{fig:language}
\end{figure}

\subsubsection{Forwards Execution of Primitives}

Given the definition of the primitives in $P$, we can define the forwards execution of the primitives in WebAssembly, as shown in \cref{fig:language}.
Non-determinism is introduced exclusively through the \textsc{input-prim} rule, which is used to evaluate input primitives.
The evaluation of the primitive $p$ non-deterministically returns a value $v$ from the codomain of the primitive function, $\lfloor p(v^*_0)_{ret} \rfloor$.
Here, the rule simply discards the compensating action, and places the return value of the primitive on the stack.
The \textsc{output-prim} rule works analogously, except the evaluation produces its return value deterministically.
Note that compensating actions $p(v^*_0)_{cps}$ are not used for the regular forward execution, but are crucial when moving backwards in time.
In the next sections, we show how the compensating actions are used during multiverse debugging.

\subsection{Configuration of the Multiverse Debugger} \label{sec:configuration}
\begin{figure*}
        \[
	\begin{array}{ l l c l }
            \emph{(Debugger state)}         & \textit{dbg}               & \Coloneqq  & \langle es, msg, mocks, K_n \; | \; S^* \rangle \\
            \emph{(Execution state)}        & es                         & \Coloneqq  & \textsc{play} \; | \; \textsc{pause} \\
            \emph{(Incoming messages)}      & msg                        & \Coloneqq  & \varnothing \; | \; \textit{step} \; | \; \textit{stepback} \; | \; \textit{pause} \; | \; \textit{play} \; | \\
                                            &                            &            & \textit{mock} \; | \; \textit{unmock} \\
            \emph{(Program state)}          & K                          & \Coloneqq  & \{ s, v^*, e^* \} \\
            \emph{(Overrides)}              & mocks                      & \Coloneqq  & \varnothing \\
                                            &                            &            & \textit{mocks}, (j,v^*) \mapsto v \\
            \emph{(A snapshot)}             & S_n                        & \Coloneqq  & \{ K_m , p_{cps} \} \\
            \emph{(Snapshots list)}         & S^*                        & \Coloneqq  & S_0 \cdot ... \cdot S_{n-1} \cdot S_n \\
            \emph{(Starting state)}         & dbg_{start}                & \Coloneqq  & \langle \textsc{pause}, \varnothing, \varnothing, K_0 \; | \; \{ K_0 , E \} \rangle \\
            \emph{(Empty action)}           & r_{nop}                    & \Coloneqq  & \lambda () . \textsf{nop} \\
        \end{array}
	\]
        \caption{The multiverse debugger state for WebAssembly with input and output primitives.}
	\label{fig:configurations}
\end{figure*}

Using the recipe for defining debugger semantics from \citet{torres19}, we can define our multiverse debugger on top of the extended WebAssembly semantics presented in \cref{sec:webassembly}.
\Cref{fig:configurations} shows the configuration of the multiverse debugger for WebAssembly with input and output primitives.
The program state in the underlying language semantics is labeled with an iteration index $n$ corresponding to the number of steps in the underlying semantic since the start of the execution, or the depth in the multiverse tree.
The debugger state $dbg$ contains the execution state of the program, incoming debug message $msg$, mocked input $mocks$, the program state $K_n$, and the snapshot list $S^*$.
The snapshots $S^*$ are a cons list of snapshots $S_n$, containing the program state $K_n$ and the compensating action $p_{cps}$.
The rules presented in the following sections show how the snapshot list is extended—and how the snapshots are used to travel back in time.

Mocked inputs are stored as a key value pairs, where the index identifying the input primitive $j$, and the list of argument values $v^*$ are mapped to the overriding return value $v$.
The key value map is represented here as a partial function, which compares lists of values $v^*$ element-wise.
For any key that is not defined in the map, we write $mocks(j,v^*) = \varepsilon$.

The starting state of the debugger $\textit{dbg}_{\textit{start}}$ is defined as the paused state with no incoming or outgoing messages, an empty mocks environment, the initial program state $K_0$, and a snapshot list containing only the initial snapshot $S_0 = \{ K_0, r_{nop} \}$.
Here $r_{nop}$ is the empty action, which takes no arguments and returns nothing.
This function indicates that no compensating action is needed.

\subsection{Forwards Exploration in Multiverse Debuggers}\label{sec:forwards}
\begin{figure}
        \begin{mathpar}
                \inferrule[(\textsc{run})]
       	            { 
                        \textsf{non-prim } K_n \\
                        K_n \hookrightarrow_{i} K_{n+1}
                    }
                    { \langle \textsc{play}, \varnothing, mocks, K_n \; | \; S^* \rangle
                      \hookrightarrow_{d,i}
                  \langle \textsc{play}, \varnothing, mocks, K_{n+1} \; | \; S^* \rangle }

                 \inferrule[(\textsc{step-forwards})]
       	            { 
                        \textsf{non-prim } K_n \\
                        K_n \hookrightarrow_{i} K_{n+1}
                    }
                    { \langle \textsc{pause}, step, mocks, K_n \; | \; S^* \rangle
                      \hookrightarrow_{d,i}
                  \langle \textsc{pause}, \varnothing, mocks, K_{n+1} \; | \; S^* \rangle }

                  \inferrule[(\textsc{pause})]
       	            { 
                    }
                    { \langle \textsc{play}, pause, mocks, K_n \; | \; S^* \rangle
                      \hookrightarrow_{d,i}
                  \langle \textsc{pause}, \varnothing, mocks, K_n \; | \; S^* \rangle }

                  \inferrule[(\textsc{play})]
       	            { 
                    }
                    { \langle \textsc{pause}, play, mocks, K_n \; | \; S^* \rangle
                      \hookrightarrow_{d,i}
                  \langle \textsc{play}, \varnothing, mocks, K_n \; | \; S^* \rangle }
	\end{mathpar}
        \caption{The small-step rules describing forwards exploration in the multiverse debugger for WebAssembly instructions without primitives.}
	\label{fig:forwards}
\end{figure}

\Cref{fig:forwards} shows the basic small-step rules for stepping forwards in the multiverse debugger, without input and output primitives.
These rules allow the debugger to explore traditional WebAssembly programs, without any non-deterministic input or output.
For clarity, we use several shorthand notations in the rules.
We use the notation $(K_n \hookrightarrow_{i} K_{n+1})$ to say that the program state $K$ takes a step to the program state $K'$ in the underlying language semantics, where $K' = K_{n+1}$.
The notation $(\textsf{non-prim } K)$ is used to indicate that the program state $K$ is not a primitive call, or more fully, it is not the case that $K = \{ s ;v^*; v^*_0 (call \; j) \} \wedge P(j) = p$.
We describe the rules in detail below.

\begin{description}
    \item[\textsc{run}] The rule for running the program forwards in the underlying language semantics. The debugger takes a step in the underlying language semantics ($K_n \hookrightarrow_{i} K_{n+1}$) as long as the execution state is \textsc{play}, there are no incoming or outgoing messages, and the program state is not a primitive call.
        While running in this way, the snapshot list $S^*$ remains unchanged.
    \item[\textsc{step-forwards}] When the debugger receives the $step$ message, it takes one more step ($K_n \hookrightarrow_{i} K_{n+1}$), and transitions to the \textsc{pause} state if it was not already paused.

    \item[\textsc{pause}] When the debugger receives a $pause$ message in the \textsc{play} state, it transitions to the \textsc{pause} state. Note that afterwards, all the run rules are no longer applicable.

    \item[\textsc{play}] The rule for continuing the execution. When the debugger receives the $play$ message in the \textsc{pause} state, the execution state transitions to the \textsc{play} state.
\end{description}

\begin{figure}
        \begin{mathpar}
                \inferrule[(\textsc{run-prim-in})]
       	            { 
                        K_n = \{ s ;v^*; v^*_0 (call \; j) \} \\
                        P(j) = p \\
                        p \in P^{In} \\
                        mocks(j, v^*_0) = \varepsilon \\
                        K_n \hookrightarrow_{i} K_{n+1} \\
                    }
                    { \langle \textsc{play}, \varnothing, mocks, K_n \; | \; S^* \rangle
                      \hookrightarrow_{d,i}
                  \langle \textsc{play}, \varnothing, mocks, K_{n+1} \; | \; S^* \cdot \{K_{n+1} , r_{nop}\} \rangle }

                \inferrule[(\textsc{run-prim-out})]
       	            { 
                        K_n = \{ s ;v^*; v^*_0 (call \; j) \} \\
                        P(j) = p \\
                        p \in P^{Out} \\
                        p(v^*_0) = \{ \textsf{ret } v, \textsf{cps } r \} \\
                        K_{n+1} = \{ s ;v^*; v \} \\
                    }
                    { \langle \textsc{play}, \varnothing, mocks, K_n \; | \; S^* \rangle
                      \hookrightarrow_{d,i}
                  \langle \textsc{play}, \varnothing, mocks, K_{n+1} \; | \; S^* \cdot \{K_{n+1} , r\} \rangle }
	\end{mathpar}
        \caption{The small-step rules describing forwards exploration for input and output primitives in the multiverse debugger for WebAssembly, without input mocking.}
	\label{fig:forwards-prim}
\end{figure}

The rules in \cref{fig:forwards-prim} are the minimal set of rules for stepping forwards when the next instruction in the program state $K_n$ is a primitive call.
The rules also describe the snapshotting behavior of the multiverse debugger, which is needed to travel back in time.
We describe the run rules in detail below.
The step rules are identical to the run rules, but they transition to the paused state after taking a step, and are triggered by the $step$ message.
These rules can be found in \cref{app:rules}.

\begin{description}
    \item[\textsc{run-prim-in}] The rule for calling an input primitive. When the program state is a primitive call, the call will not be mocked, there are no incoming messages, and the execution state is \textsc{play}, the debugger takes a step forwards in the underlying language semantics ($K_n \hookrightarrow_{i} K_{n+1}$), and adds a new snapshot to the snapshot list $\{K_{n+1} , r_{nop}\}$.
        The compensating action $r_{nop}$ is used to indicate that no compensating action is needed for input primitives.
    \item[\textsc{run-prim-out}] The rule for calling an output primitive. When the program state is an output primitive call, there are no incoming messages, and the execution state is \textsc{play}, the debugger will perform the primitive call similar to the \textsc{output-prim} rule in the underlying language semantics.
        It adds the return value of the primitive to the stack, moving the state to $K_{n+1}$, the same state reached by the underlying language semantics.
        However, it will not discard the compensating action $p(v^*_0)_{cps} = r$, but add a new snapshot to the snapshot list $\{K_{n+1} , r\}$.
\end{description}

The rules described here form the backbone of the multiverse debugger, and already allow for a multiverse debugger that can explore the execution of a program forwards in time.
It is important to note that the debugger always takes snapshots when a primitive call is made.
Only when this primitive is an output primitive, can the state of the environment change.
On the other hand, only input primitives can introduce new branches to the multiverse tree.
In the next section, we discuss how the multiverse debugger can explore the execution of a program backwards in time.

\subsection{Backwards Exploration with Checkpointing}\label{sec:backwards}

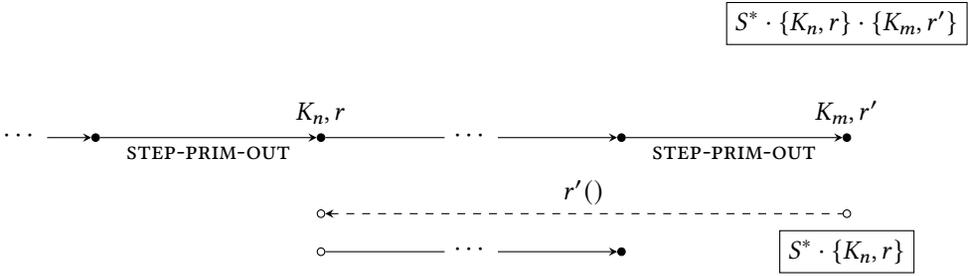
\begin{figure}
\begin{tikzpicture}
    % Nodes (Top Row)
    \node (start) at (-5,0) {$\cdots$};
    \node[draw, circle, fill=black, inner sep=1pt] (A) at (-4,0) {};
    \node[draw, circle, fill=black, inner sep=1pt, label=above:{$K_{n}, r$}] (B) at (-1,0) {};
    \node (C) at (1,0) {$\cdots$};
    \node[draw, circle, fill=black, inner sep=1pt] (D) at (3,0) {};
    \node[draw, circle, fill=black, inner sep=1pt, label=above:{$K_m, r'$}] (E) at (6,0) {};

    % Nodes (Bottom Row)
    \node[draw, circle, inner sep=1pt] (F1) at (-1,-1.00) {};
    \node[draw, circle, inner sep=1pt] (F2) at (-1,-1.5) {};
    \node[draw, circle, inner sep=1pt] (G1) at (6,-1.00) {};
    \node (M) at (1,-1.5) {$\cdots$};
    \node[draw, circle, fill=black, inner sep=1pt] (G2) at (3,-1.5) {};

    % Paths (Top Row) with custom Stealth arrows
    \draw[-farrow] (start) -- (A);
    \draw[-farrow] (A) -- node[below] {\textsc{step-prim-out}} (B);
    \draw[-] (B) -- (C);
    \draw[-farrow] (C) -- (D);
    \draw[-farrow] (D) -- node[below] {\textsc{step-prim-out}} (E);

    % Parallel paths (Bottom Row)
    \draw[farrow-, dashed] (F1) -- node[above] {$r'()$} (G1);
    \draw[-] (F2) -- (M);
    \draw[-farrow] (M) -- (G2);

    % Snapshot lists
    \node[draw, rectangle] at (6,1.5) {$S^* \cdot \{K_n , r\} \cdot \{K_m , r'\}$};
    \node[draw, rectangle] at (6,-1.5) {$S^* \cdot \{K_n , r\}$};

\end{tikzpicture}
\caption{Schematic of the \textsc{step-back-compensate} rule applied in state $K_m$. The dotted arrow shows how the debugger jumps to the previous state $K_n$, and compensates the output primitive with $r'()$, while the full arrows show the normal execution.
    Top right: the snapshots before. Bottom right: the snapshots after.}
\label{schematic:backwards}
\end{figure}

The multiverse debugger is also a time travel debugger, which means it can move backwards in time.
It does this by restoring the program state from a previous snapshot, and then replaying the program's execution from that point.
\Cref{schematic:backwards} shows schematically the way the debugger steps back in time using the snapshots.
In the situation depicted by the figure, the debugger has just stepped over an output primitive and added a snapshot to the snapshot list.
Since the debugger will now step back over the execution of this output primitive, its effects must be reversed with the compensating action $r''$ in the last snapshot.
This is shown with the dashed arrow.
As part of this jump back in time, the snapshot containing the compensating action $r''$ is removed from the snapshot list and the program state $K_n$ from the next snapshot is restored.

Since snapshots are only added when the program performs a primitive call, the second to last snapshot in the list was taken after the previous primitive call resulting in state $K_n$.
This means, that after restoring the internal virtual machine state, the program is now at the point right after the previous primitive call.
Starting from this point, the debugger can replay the program's execution forwards to $K_{m-1}$, which will not include any primitive calls.
This means the steps will be deterministic, and will not change the external environment.
This corresponds with the full arrow at the bottom of the figure.
Next to it is shown the snapshot list after the step back, which now only contains the snapshot of the state $K_n$.

In the outlined scenario the last transition in the chain, from $K_{m-1}$ to $K_m$, performs an output action.
However, if this transition is a standard WebAssembly instruction instead, no compensating action will be performed.
Instead, the debugger will immediately restore the virtual machine state to $K_n$ and replay the program's execution forwards to $K_{m-1}$.
In this case, none of the snapshots will be removed.
This enables the debugger to continue stepping back in time.

\begin{figure}
	\begin{mathpar}
                \inferrule[(\textsc{step-back})]
       	            {
                        K_n \hookrightarrow_{i}^{m-n-1} K_{m-1} \\
                        m > n \\
                    }
                    { \langle \textsc{pause}, step back, mocks, K_m \; | \; S^* \cdot \{K_{n} , r\} \} \rangle
                      \hookrightarrow_{d,i}
                  \langle \textsc{pause}, \varnothing, mocks, K_{m-1} \; | \; S^* \cdot \{K_{n} , r\} \rangle }

                \inferrule[(\textsc{step-back-compensate})]
       	            {
                        \textsf{first }  r'() \textsf{ then } K_n \hookrightarrow_{i}^{m-n-1} K_{m-1} \\
                    }
                    { \langle \textsc{pause}, step back, mocks, K_m \; | \; S^* \cdot \{K_{n} , r\} \cdot \{K_{m} , r'\} \rangle
                      \hookrightarrow_{d,i}
                  \langle \textsc{pause}, \varnothing, mocks, K_{m-1} \; | \; S^* \cdot \{K_{n} , r\} \rangle }
	\end{mathpar}
        \caption{The small-step reduction rule for stepping backwards in the multiverse debugger.}
	\label{fig:backwards}
\end{figure}

Each of the two outlined scenarios correspond with a rule in the multiverse debugger semantics, shown in \cref{fig:backwards}.
The first scenario where the effects of an output primitive is reversed, is described by the \textsc{step-back-compensate} rule.
The second scenario where the last transition is a standard WebAssembly instruction, is described by the \textsc{step-back} rule.
We describe the rules in detail below.

\begin{description}
    \item[\textsc{step-back}] The rule for stepping back in time. When the debugger receives a \emph{step back} message, the debugger restores the external state from the last snapshot in the snapshot list, which is not the current state.
        The debugger then replays the program's execution from that point to exactly one step ($K_n \hookrightarrow_{i}^{m-n-1} K_{m-1}$) before the starting state.
        Since the restored snapshot remains in the past, it is kept in the snapshot list, to allow for further backwards exploration.
    \item[\textsc{step-back-compensate}] The rule for stepping back in time when the last transition was a primitive call.
        This is always the case when the current state is $K_m$ is part of the last snapshot.
        When the debugger receives a \emph{step back} message, the debugger performs the compensating action $r'$ from the last snapshot in the snapshot list, which reversed the effects of the last primitive call.
        Then, the debugger restores the external state $K_n$ from the second to last snapshot in the snapshot list.
        The debugger then replays the program's execution from that point to exactly one step before the starting state.
        The last snapshot is removed from the snapshot list, since it now lies in the future.
\end{description}

In the case, where no primitive call has yet been made, the snapshot list contains exactly $\{K_0,r_{nop}\}$, as defined by $dbg_{start}$, which the \textsc{step-back} rule can jump to.
If the current state is $K_0$, stepping back is not possible.
Specifically, the \textsc{step-back} rule is not be applicable, since $m$ and $n$ are both zero, and the \textsc{step-back-compensate} rule requires the snapshot list to contain at least two snapshots.

\subsection{Instrumenting Non-deterministic Input in Multiverse Debuggers}\label{sec:mocking}

\begin{figure}
	\begin{mathpar}
                \inferrule[(\textsc{register-mock})]
       	            {
                        msg = mock(j, v^*, v) \\
                        P(j) = p \\
                        p \in P^{In} \\
                        v \in \lfloor p \rfloor \\
                        %p_{code} = (\textsf{func} \; a^n \rightarrow v \; \textsf{local} \; t^* e^*) \\
                        mocks' = mocks, (j,v^*) \mapsto v \\
                    }
                    { \langle rs, msg, mocks, K_n \; | \; S^* \rangle
                      \hookrightarrow_{d,i}
                      \langle rs, \varnothing, mocks', K_n \; | \; S^* \rangle }

                \inferrule[(\textsc{unregister-mock})]
       	            {
                        msg = unmock(j, v^*) \\
                        %p_{code} = (\textsf{func} \; a^n \rightarrow v \; \textsf{local} \; t^* e^*) \\
                        mocks' = mocks \setminus (j,v^*) \mapsto v \\
                    }
                    { \langle rs, msg, mocks, K_n \; | \; S^* \rangle
                      \hookrightarrow_{d,i}
                      \langle rs, \varnothing, mocks', K_n \; | \; S^* \rangle }

                \inferrule[(\textsc{step-mock})]
       	            { 
                        K_n = \{ s ;v^*; v^*_0 (call \; j) \} \\
                        P(j) = p \\
                        p \in P^{In} \\
                        mocks(j, v^*_0) = v \\
                        K'_{n+1} = \{ s';v'^*;v \} \\
                    }
                    { \langle \textsc{pause}, step, mocks, K_n \; | \; S^* \rangle
                      \hookrightarrow_{d,i}
                  \langle \textsc{pause}, \varnothing, mocks, K'_{n+1} \; | \; S^* \cdot \{K'_{n+1}, r_{nop}\} \rangle }
	\end{mathpar}
        \caption{The small-step rule for mocking input in the MIO debugger, only including the step rule. The analogous rule for when the debugger is not paused (\textsc{run-mock}) is shown in \cref{app:rules}.}
	\label{fig:mocking}
\end{figure}

In order to replay execution paths in the multiverse tree accurately, the multiverse debugger needs to be able to override the input to the program.
This is done through the \emph{mocking} of input, described by the rules in \cref{fig:mocking}.
Mocking of input happens through the key value map $mocks$ shown in \cref{fig:configurations}.
New values can be added to the map using the \textsc{register-mock} rule, and existing values can be removed using the \textsc{unregister-mock} rule.
Whenever the debugger encounters an input primitive call, it will always check the $mocks$ map for an overriding value.
If a value is found, the debugger will replace the call to the primitive with the mock value $v$. This is done by the \textsc{step-mock} rule.

\begin{description}
    \item[\textsc{register-mock}] The rule for registering a new mock value in the multiverse debugger. When the debugger receives a message $mock(j, v^*, v)$, it updates the entry for $(j,v^*)$ in the $mocks$ environment to $v$.
        If an entry already exists, the rule will override the existing value.
    \item[\textsc{unregister-mock}] The rule for unregistering a mock value in the multiverse debugger. When the debugger receives a message $unmock(j, v^*)$, it removes the mock value from the $mocks$ map.
        If no value is found in the environment, the rule will have no effect. %, and the messages will simply be removed.
    \item[\textsc{step-mock}] The \textsc{step-mock} rule for stepping forwards in the multiverse debugger when an input primitive call is encountered. If the input primitive call is found in the $mocks$ map, the debugger will replace the call with the mock value $v$.
The program state is then updated to the new program state $K_{n+1}$, and a new snapshot is added to the snapshot list.
The snapshot includes the new program state and the empty compensating action $r_{nop}$, since no compensating action is needed for input primitives.
\end{description}

\subsection{Arbitrary Exploration of the multiverse tree}\label{sec:exploration}

With the semantics of input mocking in place, we now have the entire multiverse debugger semantics for WebAssembly with input and output primitives.
In this section, we discuss how the multiverse debugger can be used to explore different universes.
This can be done by \cref{alg:jumping}.
When the debugger jumps from a state $K_m$ to a state $K_n$, the debugger will find the smallest common ancestor of $K_m$ and $K_n$, or the join.
The debugger will then step backwards from $K_m$ to the join.
We use the notation $\hookrightarrow_{r}$ to indicate that the debugger is reversing the execution, it is equivalent to a debugging step $\hookrightarrow_{d,i}$ that only uses the \textsc{step-back} and \textsc{step-back-compensate} rules.
In the final step of the algorithm, execution is replayed from the join to $K_n$ using the \textsc{step-mock} rule whenever it encounters a non-deterministic primitive call.

\begin{algorithm}[htb]
    \caption{The algorithm for traveling to any position in the multiverse tree.}
    \label{alg:jumping}
    \KwData{The current program state $K_m$, the target program state $K_n$, and the snapshot list $S^*$.}

    $K_{\text{join}} \gets \text{find\_join}(K_m, K_n)$\;
    \While{$dbg_{current}[K] \neq K_{\text{join}}$}{
        $dbg_{current} \hookrightarrow_{r} dbg_{next}$\;
        $dbg_{current} \gets dbg_{next}$\;
    }

    \While{$dbg_{current}[K] \neq K_n$}{
        $dbg_{current} \hookrightarrow_{d,i} dbg_{next}$\;
        $dbg_{current} \gets dbg_{next}$\;
    }
\end{algorithm}

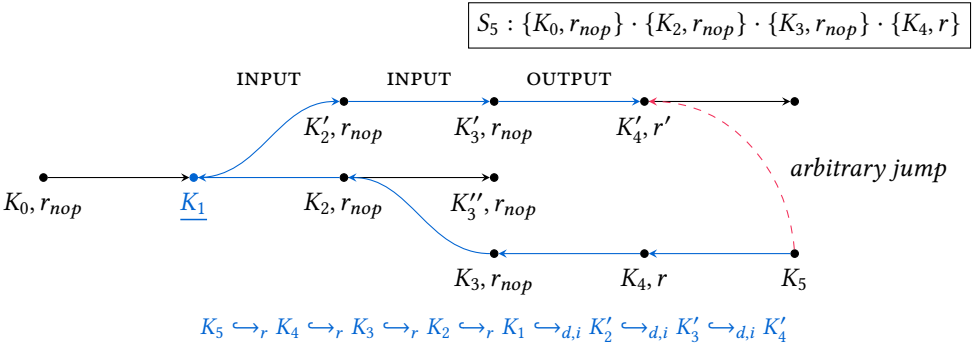
\begin{figure}
    \centering
    \begin{tikzpicture} %[label/.style={draw,rectangle,fill=white}]
        % primitives labels
        \node at (3,1.30) {\textsc{input}};
        \node at (5,1.30) {\textsc{input}};
        \node at (7,1.30) {\textsc{output}};

        % trunk
        \node[draw, circle, fill=black, inner sep=1pt] (1) at (0,0) {};
        \node[below=0.1 of 1, fill=white, text=black, inner sep=1pt] {$K_{0},r_{nop}$};

        \node[draw, circle, fill=color2, draw=color2, inner sep=1pt] (2) at (2,0) {};
        \node[below=0.1 of 2, fill=white, text=color2, inner sep=1pt] {\underline{$K_{1}$}};

        % top branch
        \node[draw, circle, fill=black, inner sep=1pt] (5) at (4,1) {};
        \node[below=0.1 of 5, fill=white, text=black, inner sep=1pt] {$K'_{2},r_{nop}$};

        \node[draw, circle, fill=black, inner sep=1pt] (6) at (6,1) {};
        \node[below=0.1 of 6, fill=white, text=black, inner sep=1pt] {$K'_{3},r_{nop}$};

        \node[draw, circle, fill=black, inner sep=1pt] (T) at (8,1) {};
        \node[below=0.1 of T, fill=white, text=black, inner sep=1pt] {$K'_{4},r'$};

        \node[draw, circle, fill=black, inner sep=1pt] (10) at (10,1) {};

        % middle branch
        \node[draw, circle, fill=black, inner sep=1pt] (3) at (4,0) {};
        \node[below=0.1 of 3, fill=white, text=black, inner sep=1pt] {$K_{2},r_{nop}$};

        \node[draw, circle, fill=black, inner sep=1pt] (4) at (6,0) {};
        \node[below=0.1 of 4, fill=white, text=black, inner sep=1pt] {$K''_{3},r_{nop}$};

        % bottom branch
        \node[draw, circle, fill=black, inner sep=1pt] (7) at (6,-1) {};
        \node[below=0.1 of 7, fill=white, text=black, inner sep=1pt] {$K_{3},r_{nop}$};

        \node[draw, circle, fill=black, inner sep=1pt] (8) at (8,-1) {};
        \node[below=0.1 of 8, fill=white, text=black, inner sep=1pt] {$K_{4},r$};

        \node[draw, circle, fill=black, inner sep=1pt] (9) at (10,-1) {};
        \node[below=0.1 of 9, fill=white, text=black, inner sep=1pt] {$K_{5}$};

        % label for snapshots
        \node[draw, align=left] at (9,2.0) {$S_5: \{ K_0 , r_{nop} \} \cdot \{K_2, r_{nop}\} \cdot \{K_3, r_{nop}\} \cdot \{K_4,r\}$};

        %\draw[-farrow] (start) -- (1);
        \draw[-farrow] (1) -- (2);
        \draw[farrow-, draw=color2] (2) -- (3);
        \draw[-farrow, draw=color2] (2) to [out=0,in=180] (5);
        \draw[-farrow, draw=color2] (5) -- (6);
        \draw[-farrow, draw=color2] (6) -- (T);
        \draw[-farrow] (T) -- (10);
        \draw[-farrow] (3) -- (4);
        \draw[farrow-, draw=color2] (3) to [out=0,in=180](7);
        \draw[farrow-, draw=color2] (7) -- (8);
        \draw[farrow-, draw=color2] (8) -- (9);

        % step back
        \draw[dashed,-farrow, draw=color0] (9) to [out=95,in=0] node[near start, above right] {\textit{arbitrary jump}} (T);

        % path
        \node at (6,-2.0) {\small \textcolor{color2}{$K_5 \hookrightarrow_{r} K_4 \hookrightarrow_{r} K_3 \hookrightarrow_{r} K_2 \hookrightarrow_{r} K_1 \hookrightarrow_{d,i} K'_2 \hookrightarrow_{d,i} K'_3 \hookrightarrow_{d,i} K'_4$}};
    \end{tikzpicture}
    \caption{Schematic of how the multiverse debugger can jump to any arbitrary state in the past, using the \textsc{step-back} and \textsc{step-mock} rules.
    For the arbitrary jump from state $K_5$ to $K_4'$, the join $K_1$ is underlined and shown in blue.
    Top right: the list of snapshots before the arbitrary jump. Bottom: the execution path from $K_5$ to $K_4'$.
    Steps with the \textsc{step-back} and \textsc{step-back-compensate} rules are shown as $\hookrightarrow_r$.}
    \label{schematic:arbitrary-jump}
\end{figure}

\Cref{schematic:arbitrary-jump} illustrates the algorithm for jumping to an arbitrary state, when the user clicks on a node on another branch in the multiverse tree.
The figure shows a possible multiverse tree for a program where the second and third instruction are input primitives.
The program has executed two input primitives in a row, and the debugger has explored some of the possible inputs.
Each node in the figure is labeled with the program state and possible compensating action, where $r_{nop}$ indicates that no compensating action is needed.
For clarity, the external state are also numbered.
The figure shows clearly that the external state only changes after a primitive call.
The current state is $K_5$, and the debugger wants to jump to $K_4'$.
Per the algorithm, the debugger finds the join of the two states, which is $K_1$.
The debugger then replays the execution from $K_5$ to $K_1$ in reverse order, using the \textsc{step-back} and \textsc{step-back-compensate} rules.
It is important that the debugger steps back one instruction at a time, to ensure that the external state is correctly restored.
From the join $K_1$, the debugger replays the execution to $K_4'$ in the forward order, using the \textsc{step-mock} rule whenever it encounters a non-deterministic primitive call.
This ensures that the jump deterministically follows the exact execution path, thereby ensuring that the external state is correctly restored.

\subsection{Correctness of the Multiverse Debugger Semantics}\label{sec:correctness}

Given the small-step semantics, we can prove the correctness of the MIO debugger in terms of soundness and completeness.
The soundness theorem states that for any debugging session ending in a certain state, there also exists a forwards execution path in the underlying language semantics to that state.
A debugging session is seen as any number of debugging steps starting from the initial debugging state $dbg_{start}$.
The completeness of the debugger means that the debugger can always find a path in the multiverse tree that corresponds to a path in the underlying language semantics.
Together, these properties ensure that the debugger is correct in terms of its observation of the underlying language, and will never observe any inaccessible states.
For brevity, we only provide a sketch of the proofs here, but the full proofs can be found in \cref{app:proofs}.

\newcommand{\theoremdebuggersoundness}{%
    Let $K_0$ be the start WebAssembly configuration, and $dbg$ the debugging configuration containing the WebAssembly configuration $K_n$.
    Let the debugger steps $\hookrightarrow^*_{d,i}$ be the result of a series of debugging messages.
    Then:
    $$\forall dbg : dbg_{start} \hookrightarrow^*_{d,i} dbg \Rightarrow K_{0} \hookrightarrow^*_i K_n$$
}

\begin{theorem}[Debugger soundness]\label{theorem:debugger-soundness}
    \theoremdebuggersoundness
\end{theorem}

The proof for debugger soundness proceeds by induction over the number of steps in the debugging session.
In the base case, where the debugging session consists of a single step, the proof is trivial since the step starts from the initial state.
In the inductive case, the proof proceeds very similarly, the only non-trivial cases are those for stepping backwards and mocking.

\newcommand{\theoremdebuggercompleteness}{%
    Let $K_{0}$ be the start WebAssembly configuration for which there exists a series of transition $\hookrightarrow^*_i$ to another configuration $K_n$. Let the debugging configuration with $K_n$ be dbg.
    Then:
    $$\forall K_n : K_{0} \hookrightarrow^*_i K_n \Rightarrow dbg_{start} \hookrightarrow^*_{d,i} dbg$$
}

\begin{theorem}[Debugger completeness]\label{theorem:debugger-completeness}
    \theoremdebuggercompleteness
\end{theorem}

The proof for completeness follows almost directly from the fact that for every transition in the underlying language semantics, the debugger can take a corresponding step. For non-deterministic input primitives, we can step to the same state with the \textsc{register-mock} and \textsc{step-mock} rules.

Together the debugger soundness and completeness theorems ensure that the multiverse debugger is correct in terms of its observation of the underlying language semantics.
However, it gives us no guarantees about the correctness of the compensating actions, and the consistency of external effects during a debugging session.
Due to the way effects on the external environment are presented in the MIO debugger semantics, we can define the entire effect of a debugging session of regular execution, both as ordered lists of steps that have external effects.
There are only two options, the output primitive rules, and the rule that applies the compensating action.

\begin{definition}[External state effects]
    The function $external$ returns the steps affecting external state for any series of rules in the debugging or underlying language semantics.
    $$external(p) = \left\{\begin{array}{ll}
            (s \text{ for } s \text{ in } p \text{ where } s = \textsc{step-prim-out} & \text{if } p = dbg \hookrightarrow^*_{d,i} dbg' \\
                                        \vee \; s = \textsc{ step-back-compensate})      & \\
            (s \text{ for } s \text{ in } p \text{ where } s = \textsc{output-prim} ) & \text{if } p = K \hookrightarrow^*_{i} K' \\
    \end{array}\right.
    $$
\end{definition}

Using this definition, we can prove that the external effects of any debugging session ending in a certain state, are the same as the effects of the regular execution of the program ending in that same state.
The definition for the equivalence of external effects ($\equiv$) is given in \cref{app:proofs}.

\newcommand{\theoremcompensatesoundness}{%
        Let $K_0$ be the start WebAssembly configuration, and $dbg$ the debugging configuration containing the WebAssembly configuration $K_n$.
    Let the debugger steps $\hookrightarrow^*_{d,i}$ be the result of a series of debugging messages.
    Then:
    $$\forall dbg : external(dbg_{start} \hookrightarrow^*_{d,i} dbg) \equiv external(K_{0} \hookrightarrow^*_i K_n)$$
}

\begin{theorem}[Compensation soundness]\label{theorem:compensate-soundness}
    \theoremcompensatesoundness
\end{theorem}

The proof of this theorem is based on the fact that our multiverse debugger is a rooted acyclic graph, and a debugging session is a walk in this tree starting from the root, which can include the same edge several times.
Any such walk in a tree can be constructed by adding any number of random closed walks to the path from the root to the final node.
Such closed walks are null operations in terms of their effect on the external state.
This leaves only the forward steps of the minimal path to be considered, meaning the external effects of a debugging session are always the same as those of the regular execution of the program.

\section{The MIO debugger}\label{sec:implementation}

\label{sec:warduino}
We have implemented the multiverse debugger described above in a prototype debugger, called the MIO debugger.
The MIO debugger is built on top of a WebAssembly runtime for microcontrollers, called WARDuino \cite{lauwaerts24a}.
WARDuino is written in C++, includes primitives for controlling hardware peripherals, and has support for traditional remote debugging.  
Our prototype implementation builds further on its virtual machine and the remote debugging facilities but needed to be extended significantly in order to support all the basic operations for multiverse debugging: smart snapshotting, mocking of primitives and reversible actions.  
Additionally, we created a high-level interface which implements the message passing interface described in \cref{sec:multiverse-debugger} as messages in the remote debugger of WARDuino. On top of this interface we built a Kotlin application for debugging AssemblyScript programs on microcontrollers running WARDuino.
This application keeps track of the program states, and shows them as part of the multiverse tree, as shown in \cref{fig:multiverse-debugger} from \cref{sec:practice}.
\Cref{alg:jumping} for arbitrary jumping, is implemented at the level of the Kotlin application using the message passing interface of the remote debugger.
Finally, the MIO prototype also has support for expressing simple dependencies, so that the mocking of the various sensor values can be limited depending on state of the output pins. 

\subsection{Output: Reversible Primitives}\label{sec:rotate}

Primitives in WARDuino are implemented in the virtual machine using C macros.
In order to implement reversible primitives, we have extended the existing macros with two new macros; one defines how the external state effected by the primitive can be captured, and the other defines the compensating action given this captured state.
When stepping back over a primitive, the compensating action looks at the state captured after the previous primitive call, and restores this external state.
This is the same as undoing the effects of the last primitive call.

\definecolor{mRed}{rgb}{1,0.4,0.4}
\definecolor{mBlue}{rgb}{0.4,0.4,1}
\definecolor{mGreenGraph}{RGB}{102,255,110}
\definecolor{mGreen}{rgb}{0,0.6,0}
\definecolor{mGray}{rgb}{0.5,0.5,0.5}
\definecolor{mDarkGray}{rgb}{0.2,0.2,0.4}
\definecolor{mPurple}{rgb}{0.58,0,0.82}
\definecolor{backgroundColour}{rgb}{1,1,1}

\lstdefinestyle{CStyle}%
{     basicstyle=\footnotesize\ttfamily\linespread{0.7}%
, captionpos=b%
, identifierstyle=%
, backgroundcolor=\color{backgroundColour}
, commentstyle=\color{mGray}
, keywordstyle=\bfseries\color{black}
, numberstyle=\tiny\ttfamily\color{mGray}
, stringstyle=\color{mRed}
, keywordstyle=\color{mBlue}\bfseries% reserved keywords
, keywordstyle=[2]\color{mRed}% traits
, keywordstyle=[3]\color{mBlue}% primitive types
, keywordstyle=[4]\color{mRed}% type and value constructors
, keywordstyle=[5]\color{mBlue}% macros
, columns=spaceflexible%
, keepspaces=true%
, showspaces=false%
, showtabs=false%
, showstringspaces=false%
, numbers=left%
, numbersep=5pt%
}

\begin{figure}
    \begin{minipage}[t]{.44\textwidth}
        \begin{lstlisting}[language=C++, style=CStyle,escapechar=',xleftmargin=4mm]
def_prim(rotate, threeToNoneU32) {
  int32_t speed = arg0.int32;
  int32_t degrees = arg1.int32;
  int32_t motor = arg2.int32;
  pop_args(3);
  auto encoder = encoders[motor];
  encoder->set_angle('\label{line:encode}'
    encoder->get_angle() + degrees'\label{line:relative}'
  );
  return drive(motor, encoder, speed);'\label{line:drive}'
}\end{lstlisting}
    \end{minipage}
    \hfill
    \begin{minipage}[t]{.50\textwidth}
        \begin{lstlisting}[language=C++, style=CStyle,escapechar=',firstnumber=10]
def_prim_serialize(rotate) {
  for (int m = 0; m < MOTORS; i++) {
    external_states.push_back(
    new MotorState(m, encoders[m]->angle()));'\label{line:serialize}'
}}

def_prim_reverse(rotate) {
  for (IOState s : external_states) {
    if (isMotorState(s)) {
      int motor = stoi(s.key);
      auto encoder = encoders[motor];
      encoder->set_angle(s.degrees);'\label{line:set-angle}'
      drive(motor, encoder, STD_SPEED);
}}}\end{lstlisting}
    \end{minipage}
    \caption{\emph{Left:} The implementation of the \emph{rotate} primitive. \emph{Right:} The implementation of the compensating action for the \emph{rotate} primitive, in the MIO debugger.}
    \label[listing]{fig:motor-impl}
\end{figure}

To illustrate the implementation of reversible primitives, we will use the example of the \emph{rotate} primitive, which rotates a servo motor for a given number of degrees.
The forwards implementation is shown on the left side of \cref{fig:motor-impl}.
%The servo motors are controlled by pulse-width modulation (PWM) signals.
To move the motor a given number of degrees the primitive first sets the target angle of the motor encoder, this happens on line~\ref{line:encode}.
The motor encoder is used to track the current motor angle, as well as the absolute target angle, which can be set with the \emph{set\_angle} method.
To rotate the motor a number of degrees relative to its current position, the primitive adds the degrees to the current motor angle (line~\ref{line:relative}).
Once the target angle is set, the primitive drives the motor to that angle using the \emph{drive} method, as shown on line~\ref{line:drive}.

The implementation of the compensating action for the \emph{rotate} primitive is shown on the right side of \cref{fig:motor-impl}.
First, the \emph{def\_prim\_serialize} macro captures the external state.
For each motor, the current angle of the motor is stored along with its index, as shown on line~\ref{line:serialize}.
Second, the \emph{def\_prim\_reverse} macro compensates the primitive by moving all motors back to the angles captured in the previous snapshot.
The angles captured by the \emph{def\_prim\_serialize} macro are absolute target angles. The compensating action moves the motors back to these angles by first setting the target angle, as shown on line~\ref{line:set-angle}.
It then uses the same \emph{drive} function to move the motor.

\subsection{Input: Mocking of Primitives}

The input mocking is implemented analogous to the debugger semantics, by adding a map to the in the virtual machine state.
This map is used to store the mocked values for the input primitives, which are received by a new debug message in the remote debugger.
In line with the semantics, there is also a new debug message to remove a mocked value from the map.
Currently, the map only supports registering primitive calls with their first argument.
This is sufficient for the current input primitives to be mocked, without any changes to their implementation.

The virtual machine will check the map of mocked values for every primitive call.
The prototype includes two input primitives that can be mocked in this way, the \emph{digitalRead} primitive which reads the value of a digital pin, and the \emph{colorSensor} primitive which reads a value from a uart color sensor.
The digitalRead primitive enables the user to mock the value of a digital pin, and thereby the behavior of a wide range of possible peripherals.
However, the range of possible input values is not always known statically, as it may be influenced by the output effects of the program.
To handle this, the MIO debugger includes initial support for predictable dependencies that can be defined as simple conditions, for example, \emph{"when the value of a digital pin $n$ is $x$, then input primitive $p$ with arguments $m$ will return the value $c$"}.

\subsection{Performance: Checkpointing}

To reduce memory usage, the MIO debugger only stores snapshots at certain checkpoints.
The semantics of MIO only takes snapshots after a primitive call, the prototype implementation follows this checkpointing policy precisely.
As shown by the debugger semantics and the proof, this is the minimum number of snapshots needed to enable backwards and forwards exploration of the multiverse tree.
To further reduce the performance impact on the microcontroller, snapshots are received and tracked by the desktop frontend of the MIO debugger.
To have minimal traffic between the backend and frontend, snapshots after primitive calls are sent automatically to the frontend.
Alternatively, the frontend can request snapshots at will through the remote debugger interface.

\section{Evaluation}\label{sec:evaluation}

To validate that our checkpointing strategy is performant enough for apply multiverse debugging on microcontrollers we performed a number of experiments. 
All experiments were performed on an STM32L496ZG microcontroller running at 80 MHz.
This microcontroller was connected to a laptop running the MIO debugger frontend that communicates with the microcontroller.

\begin{figure}
    \centering
    \begin{tikzpicture}
        \begin{groupplot}[
            group style={
            group size=2 by 1, % Two plots side by side
            horizontal sep=0.06\textwidth, % Space between plots
            },
            width=0.42\textwidth, % Width of each plot
            xlabel={Instructions executed}, 
            ylabel={Overhead (rel. \emph{no snapshotting})},
            table/col sep=comma,
            xmin=250, xmax=1250,
            xtick={250, 500, 750, 1000, 1250},
            tick label style={font=\scriptsize},
            label style={font=\small},
            ymin=0,
            scale only axis,
            grid=major,
            legend to name={sharedlegend},
            legend columns=3,
            legend style={font=\small}
            ]

            % First plot with all policies
            \nextgroupplot[
            ]
            % Define the baseline (x-axis at y=0) for the fill
            \path[name path=horizon] (axis cs:250,0) -- (axis cs:1250,0);

            \addplot[
            color=color0,
            mark=*,
            mark options={scale=0.7, fill=color0, draw=none},
            smooth,
            name path=1,
            ] 
            table[
            x={Instructions executed},
            y expr=\thisrow{Overhead},
            col sep=comma,
            restrict expr to domain={\thisrow{Interval}}{0:0}
            ] {benchmarks/forward-execution-checkpointing.csv};

            \addplot[
            color=color1,
            mark=*,
            mark options={scale=0.7, fill=color1, draw=none},
            smooth,
            name path=2,
            ] 
            table[
            x={Instructions executed},
            y expr=\thisrow{Overhead},
            col sep=comma,
            restrict expr to domain={\thisrow{Interval}}{1:1}
            ] {benchmarks/forward-execution-checkpointing.csv};
            \addlegendentry{every instruction}

            \addplot[
            color=color2,
            mark=*,
            mark options={scale=0.7, fill=color2, draw=none},
            smooth,
            name path=3,
            ] 
            table[
            x={Instructions executed},
            y expr=\thisrow{Overhead},
            col sep=comma,
            restrict expr to domain={\thisrow{Interval}}{5:5}
            ] {benchmarks/forward-execution-checkpointing.csv};

            \addplot[
            color=color3,
            mark=*,
            mark options={scale=0.7, fill=color3, draw=none},
            smooth,
            name path=4,
            ] 
            table[
            x={Instructions executed},
            y expr=\thisrow{Overhead},
            col sep=comma,
            restrict expr to domain={\thisrow{Interval}}{10:10}
            ] {benchmarks/forward-execution-checkpointing.csv};

            \addplot[
            color=color4,
            mark=*,
            mark options={scale=0.7, fill=color4, draw=none},
            smooth,
            name path=5,
            ] 
            table[
            x={Instructions executed},
            y expr=\thisrow{Overhead},
            col sep=comma,
            restrict expr to domain={\thisrow{Interval}}{50:50}
            ] {benchmarks/forward-execution-checkpointing.csv};

            \addplot[
            color=color5,
            mark=*,
            mark options={scale=0.7, fill=color5, draw=none},
            smooth,
            name path=6,
            ] 
            table[
            x={Instructions executed},
            y expr=\thisrow{Overhead},
            col sep=comma,
            restrict expr to domain={\thisrow{Interval}}{100:100}
            ] {benchmarks/forward-execution-checkpointing.csv};

            \addplot [
            color0,
            fill opacity=0.6,
            smooth,
            forget plot,       % No legend entry for the fill
            ] fill between [
            of = {1 and horizon},
            ];

            \addplot [
            color1,
            fill opacity=0.6,
            smooth,
            forget plot,       % No legend entry for the fill
            ] fill between [
            of = {2 and 3},
            ];

            \addplot [
            color2,
            fill opacity=0.6,
            smooth,
            forget plot,       % No legend entry for the fill
            ] fill between [
            of = {3 and 4},
            ];

            \addplot [
            color3,
            fill opacity=0.4,
            smooth,
            forget plot,       % No legend entry for the fill
            ] fill between [
            of = {4 and 5},
            ];

            \addplot [
            color4,
            fill opacity=0.4,
            smooth,
            forget plot,       % No legend entry for the fill
            ] fill between [
            of = {5 and 6},
            ];

            \coordinate (a) at (axis cs:1250,0);
            \coordinate (b) at (axis cs:1250,20);

            % Second plot (excluding "Snapshot at every instruction")
            \nextgroupplot[
            ylabel={},
            ]
            % Define the baseline (x-axis at y=0) for the fill
            \path[name path=horizon] (axis cs:250,0) -- (axis cs:1250,0);

            \addplot[
            color=color0,
            mark=*,
            mark options={scale=0.7, fill=color0, draw=none},
            smooth,
            name path=7,
            ] 
            table[
            x={Instructions executed},
            y expr=\thisrow{Overhead},
            col sep=comma,
            restrict expr to domain={\thisrow{Interval}}{0:0}
            ] {benchmarks/forward-execution-checkpointing.csv};
            \addlegendentry{no snapshotting}

            \addplot[
            color=color1,
            mark=*,
            mark options={scale=0.7, fill=color1, draw=none},
            smooth,
            name path=2,
            ] coordinates {(0,0) (0,0)};
            \addlegendentry{every instruction}

            \addplot[
            color=color2,
            mark=*,
            mark options={scale=0.7, fill=color2, draw=none},
            smooth,
            name path=8,
            ] 
            table[
            x={Instructions executed},
            y expr=\thisrow{Overhead},
            col sep=comma,
            restrict expr to domain={\thisrow{Interval}}{5:5}
            ] {benchmarks/forward-execution-checkpointing.csv};
            \addlegendentry{5 instructions}

            \addplot[
            color=color3,
            mark=*,
            mark options={scale=0.7, fill=color3, draw=none},
            smooth,
            name path=9,
            ] 
            table[
            x={Instructions executed},
            y expr=\thisrow{Overhead},
            col sep=comma,
            restrict expr to domain={\thisrow{Interval}}{10:10}
            ] {benchmarks/forward-execution-checkpointing.csv};
            \addlegendentry{10 instructions}

            \addplot[
            color=color4,
            mark=*,
            mark options={scale=0.7, fill=color4, draw=none},
            smooth,
            name path=10,
            ] 
            table[
            x={Instructions executed},
            y expr=\thisrow{Overhead},
            col sep=comma,
            restrict expr to domain={\thisrow{Interval}}{50:50}
            ] {benchmarks/forward-execution-checkpointing.csv};
            \addlegendentry{50 instructions}

            \addplot[
            color=color5,
            mark=*,
            mark options={scale=0.7, fill=color5, draw=none},
            smooth,
            name path=11,
            ] 
            table[
            x={Instructions executed},
            y expr=\thisrow{Overhead},
            col sep=comma,
            restrict expr to domain={\thisrow{Interval}}{100:100}
            ] {benchmarks/forward-execution-checkpointing.csv};
            \addlegendentry{100 instructions}

            \addplot [
            color0,
            fill opacity=0.6,
            smooth,
            forget plot,       % No legend entry for the fill
            ] fill between [
            of = {7 and horizon},
            ];

            \addplot [
            color2,
            fill opacity=0.6,
            smooth,
            forget plot,       % No legend entry for the fill
            ] fill between [
            of = {8 and 9},
            ];

            \addplot [
            color3,
            fill opacity=0.6,
            smooth,
            forget plot,       % No legend entry for the fill
            ] fill between [
            of = {9 and 10},
            ];

            \addplot [
            color4,
            fill opacity=0.4,
            smooth,
            forget plot,       % No legend entry for the fill
            ] fill between [
            of = {10 and 11},
            ];

            \addplot [
            color5,
            fill opacity=0.4,
            smooth,
            forget plot,       % No legend entry for the fill
            ] fill between [
            of = {11 and 7},
            ];

            \coordinate (aa) at (axis cs:250,0);
            \coordinate (bb) at (axis cs:250,19.70);
        \end{groupplot}

        \coordinate (c) at (current bounding box.north);
        \node at (c) [anchor=south] {\pgfplotslegendfromname{sharedlegend}};

        \draw [black!60, dotted, shorten <=1pt] (a) -- (aa);
        \draw [black!60, dotted, shorten >=1.5pt,shorten <=1.7pt] (b) -- (bb);

    \end{tikzpicture}
    \caption{Comparison of execution time of snapshotting: \textit{never}, at \textit{every instruction}, and checkpointing intervals; \textit{5, 10, 50, and 100 instructions}.
        The performance overhead is shown as execution time relative to the execution time when taking no snapshots.
    Left: Comparison of all policies.
    Right: The performance of the policies excluding snapshotting every instruction (zoom-in). The averages are taken over 10 runs of the same program.}
    \label{fig:snapshotting-performance}
\end{figure}
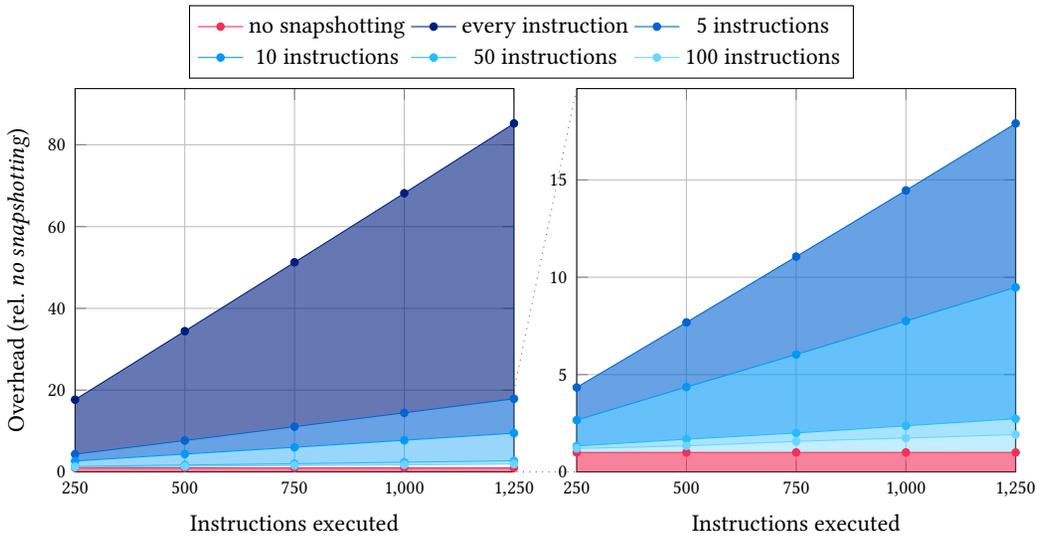

\subsection{Forward Execution with Checkpointing}

The first experiment evaluates the performance impact of checkpointing on the execution speed. % overhead of taking snapshots at different intervals in comparison to taking snapshots at every instruction or taking no snapshots at all.
We measured the execution time of a fixed number of instructions, when taking no snapshots, taking a snapshot every instruction, and for snapshotting after different intervals (5, 10, 50, or 100 instructions), as shown in \cref{fig:snapshotting-performance}. % for various different snapshot policies.
To reduce the impact of variable unknown factors, the program executed by the virtual machine includes no primitive calls. Specifically, this program checks for each integer from 1 to 13,374,242 if they are prime. Because this program has no primitive calls, the VM will only take snapshots at fixed intervals which are determined by the frontend.

The left plot shown in \cref{fig:snapshotting-performance}, gives the time it took to execute up to 1250 instructions for each snapshot policy relative to taking no snapshots.
Since snapshotting at every instruction is so much slower, we show the same zoomed in results in the right plot, but omit the results for snapshotting at every instruction.
For such small numbers of instructions, the execution time without any debugger intervention, remains roughly the same, taking on average 222.7ms. These results are shown in red.
In contrast, when taking snapshots after every executed instruction, the execution time increases dramatically.
For 1250 instructions it takes on average 19 seconds, which is around 85 times slower. % than without snapshotting.
For only 250 instructions the execution time increases seventeen-fold, to 3.9 seconds.

Once checkpointing is used the overhead reduces significantly.
When taking snapshots every five instructions, the virtual machine only needs 4 seconds to execute 1250 instructions.
Per instruction this results in an execution that is only 17.9 times slower.
Taking a snapshot every 10 instructions results in a total execution time of 2.1 seconds.
This results in a slowdown of factor 9.5.
When taking snapshots every 50 instructions, the slowdown lowers to a factor of 2.7.
Going up to a hundred instructions every snapshot, this becomes only a factor 1.9.

This initial benchmark of the checkpointing strategy shows that the performance overhead can be greatly reduced by reducing the number of snapshots taken.
Yet, execution times are still significantly slower than without any snapshotting.
This is due to the fact that the current prototype has not yet been optimized for performance.
The prototype only uses simple run-length encoding of WebAssembly memory to reduce the size of the snapshots.
In future, the snapshot sizes could be reduced greatly by only communicating the changes compared to the previous snapshot.
%This shows that, as the time between checkpoints is increased, the slowdown can be reduced significantly. %to around a factor 2 which is much better than the original 100 times slower execution.
However, in practice the performance is already sufficient to provide a responsive debugger interface as we illustrate in the online demo videos, which can be found \href{https://youtube.com/playlist?list=PLaz61XuoBNYVcQqHMAAXQNf8fz5IAMahe&si=HNrKY9YzqDFadATN}{here}\footnote{Full link: \href{https://youtube.com/playlist?list=PLaz61XuoBNYVcQqHMAAXQNf8fz5IAMahe&si=HNrKY9YzqDFadATN}{https://youtube.com/playlist?list=PLaz61XuoBNYVcQqHMAAXQNf8fz5IAMahe\&si=HNrKY9YzqDFadATN}}.
The example we highlight later in \cref{sec:usecase}, requires on average snapshot every 37 instructions.
This reduces overhead sufficiently to have a responsive debugging experience for users.
Additionally, the I/O operations by comparison typically take much longer to execute, a single action easily taking several seconds.

\begin{figure}
    \centering
    \begin{tikzpicture}
        \begin{axis}[
            xlabel={Amount of re-executed instructions},
            ylabel={Average time to step back (ms)},
            table/col sep=comma,
            xmin=0, xmax=30000,
            ymin=0, ymax=1200,
            xtick={0,5000,10000,15000,20000,25000,30000},
            axis lines=box,
            xtick pos=bottom,
            ytick pos=left,
            axis line style={-},
            tick align=inside,
            scaled x ticks=base 10:-3,
            ytick={0,100,200,300,400,500,600,700,800,900,1000},
            legend pos=south east,
            grid=both,
            height=6.5cm,
            ]

            % Define the baseline (x-axis at y=0) for the fill
            \path[name path=axis] (axis cs:0,0) -- (axis cs:30000,0);

            \addplot[
            color=color2,
            mark=*,
            mark options={scale=0.7, fill=color2, draw=none},
            name path=chart,
            ]
            table[
                x=t, y=avg_time
            ] {benchmarks/step-back-reexecute.csv};

            \addplot [
            color2,
            fill opacity=0.4,
            smooth,
            forget plot,       % No legend entry for the fill
            ] fill between [
            of = {chart and axis},
            ];

            \addplot[
            dashed,
            color=color2,
            forget plot
            ]
            coordinates {(0, 468.0) (30000, 468.0)};

            \node at (axis cs:28500, 468) [anchor=south] {\scriptsize 468 ms};
        \end{axis}
    \end{tikzpicture}
    \caption{Plot showing the average time to step back as the number of instructions requiring re-execution increases in increments of one thousand. Averages are calculated over 10 runs of the same program.}
    \label{fig:stepping-back-performance}
\end{figure}
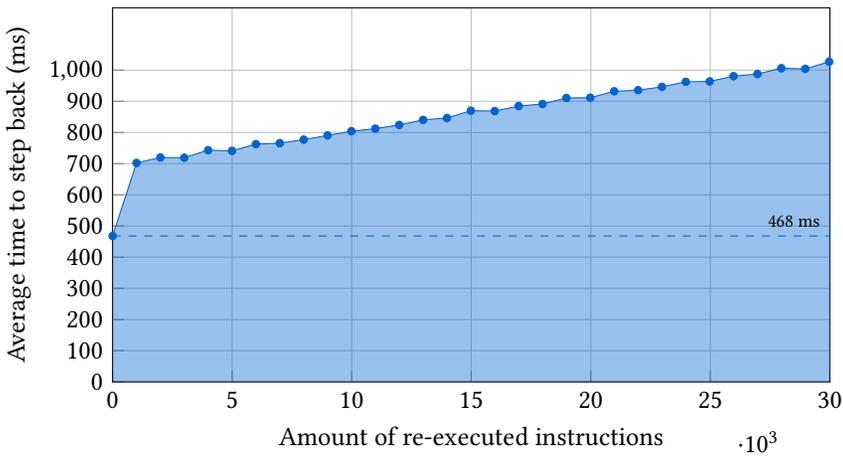

\subsection{Backwards execution}

The graphs in \cref{fig:snapshotting-performance} only show part of the picture, where the less snapshots are taken, the better the performance.
Unfortunately, there is no such thing as a free lunch, and while only taking one snapshot at the start of the program and never again, would result in the lowest possible overhead for forwards execution, this is not the case for backwards execution.
In that case stepping back would always have to re-execute the entire program.
Clearly, the further apart the snapshots, the longer it will take to step back.
To illustrate this trade-off, we examined the impact of the number of re-executed instructions on stepping back speed.

\Cref{fig:stepping-back-performance} shows the average time it takes to step back as the number of instructions requiring re-execution increases in increments of one thousand.
The averages are calculated over 10 runs of the same program used in the previous section.
When executing only a handful of instructions, the time to step back is dominated by the communication latency between the microcontroller and the debugger frontend.
On average, this results in a minimal time of 468ms to step back.
Between one thousand and 30 thousand re-executed instructions, the time to step back increases linearly by roughly 11ms per a thousand instructions.

Our analysis of the checkpoint strategy's impact on stepping back shows that the overhead is minimal.
The prototype is able to re-execute 30 thousand non-I/O instructions in around one second.
Compared to the overhead of checkpointing on forwards execution (see \cref{fig:snapshotting-performance}), we can safely conclude that in practice the overhead on backwards execution is negligible. % introduced by reducing the number of snapshots 
This is further evidenced in our \href{https://youtube.com/playlist?list=PLaz61XuoBNYVcQqHMAAXQNf8fz5IAMahe&si=HNrKY9YzqDFadATN}{demo videos}, where developers mostly have to wait for physical I/O actions to complete, and stepping back is otherwise instantaneous.

\subsection{Use case: Lego Mindstorms Color Dial}\label{sec:usecase}

To illustrate the practical potential of MIO and its new debugger approach, we present a simple reversible robot application using Lego Mindstorms components.
However, not just microcontroller applications may benefit from our novel approach, there are many application domains where output is entirely in the form of digital graphics, which are more easily reversible---such as video games, simulations, etc.
Nevertheless, to highlight the potential of the approach we demonstrate the MIO debugger using small physical robots and other microcontroller applications, as this is a more challenging environment for multiverse debugging.
Using the digital input and motor primitives described in \cref{sec:implementation}, we developed a color dial, as a simplified application.
We developed this example alongside a few others to further demonstrate the usability of the MIO debugger\footnote{Code for all examples can be found \emph{[link to repository removed for double-blind review]}}, and have created demo videos for a few of the examples, which can be found \href{https://youtube.com/playlist?list=PLaz61XuoBNYVcQqHMAAXQNf8fz5IAMahe&si=HNrKY9YzqDFadATN}{online}\footnote{Full link: \href{https://youtube.com/playlist?list=PLaz61XuoBNYVcQqHMAAXQNf8fz5IAMahe&si=HNrKY9YzqDFadATN}{https://youtube.com/playlist?list=PLaz61XuoBNYVcQqHMAAXQNf8fz5IAMahe\&si=HNrKY9YzqDFadATN}}.

\lstdefinelanguage{AssemblyScript}{
sensitive,
morecomment=[l]{//},
morekeywords={import, export, let, const, while, class, function, as, from, enum},
morekeywords=[2]{true, void, u32, i32, boolean, Pin, Color, string, Options},
morekeywords=[3]{@external},
morestring=[b]{"},
morestring=[b]` % Interpolation strings.
}

The color dial application works as follows, the robot has a color sensor that can detect the color of objects.
Depending on the color seen by the sensor, a single motor will move the needle on the dial to the location indicating the color seen by the sensor.
We built the dial using LEGO Mindstorms components~\cite{ferreira24} as shown on the left of \Cref{fig:robot}.
The right-hand side of \cref{fig:robot} shows the infinite loop that controls the robot, written in AssemblyScript.
In this loop, the robot will continually read sensor values from the color sensor. While doing so it will move the needle of the dial to the correct position indicating the current color seen by the sensor.
The needle is only moved if the color sensor sees a value different from what the dial is currently indicating.
The relative amount that the needle needs to move is calculated by taking the difference between the current color the needle is pointing at and the new color.

\begin{figure}
    \begin{minipage}{.40\textwidth}
    \centering
    \includegraphics[height=5.4cm]{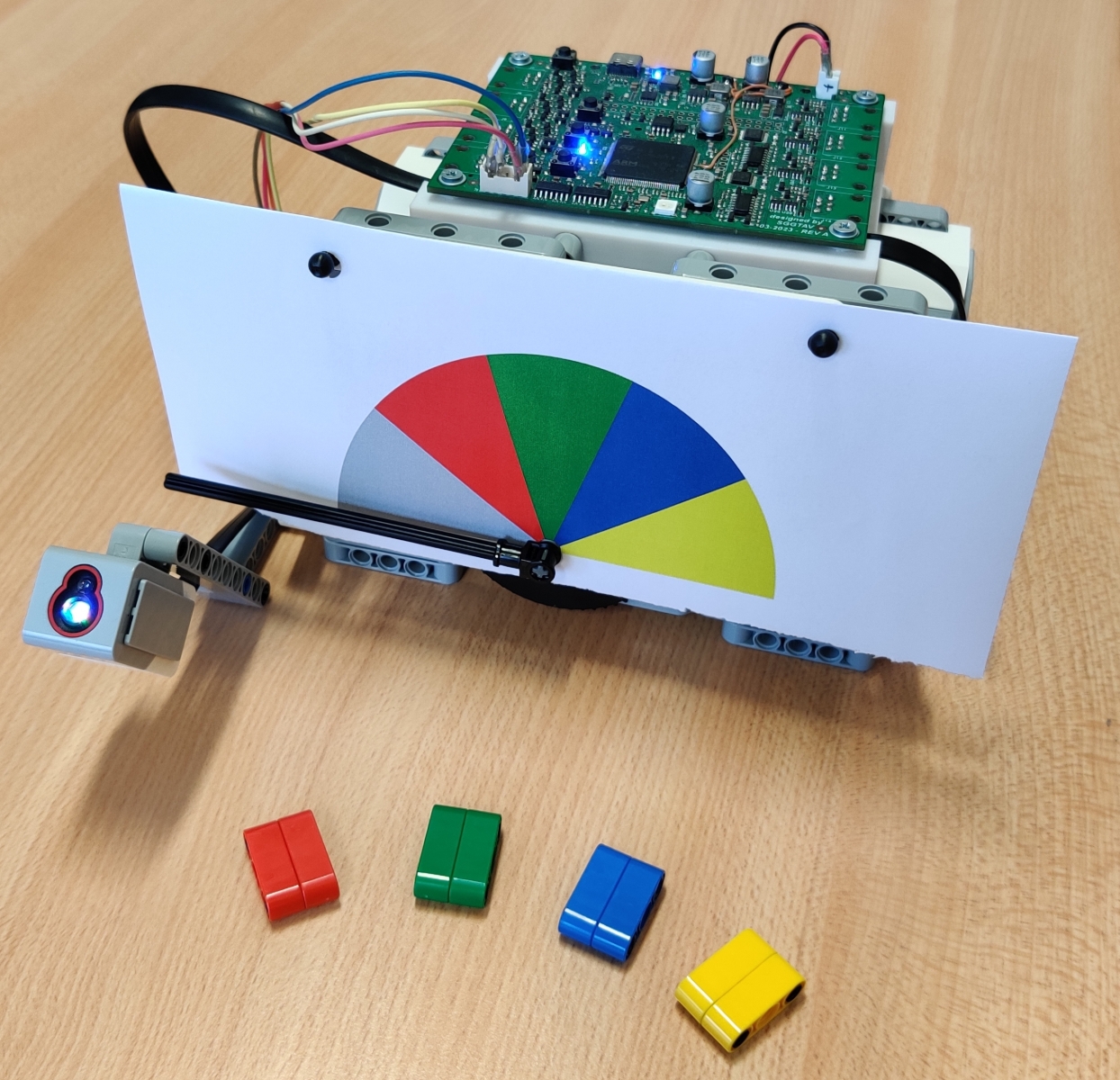}
    \end{minipage}%
\hfill%
\begin{minipage}{.54\textwidth}
        \begin{lstlisting}[language=AssemblyScript, style=CStyle,escapechar=']
enum Color { none = 0, red = 1, green = 2,
             blue = 3, yellow = 4 }

const sensor = colorSensor(Pin.IO2);
let current: Color = Color.none;
while (true) {
  let next: Color = sensor.read();
  if (next != current) {'\label{line:target}'
    // turn the needle if the color changed
    rotate(Pin.IO1,
      (next - current) * angle, speed);
  }
  current = next;
}
\end{lstlisting}
\end{minipage}
    \caption{Left: Lego color dial that recognizes the color of objects. Right: The main loop controlling the color dial. The dial has a single motor connected to pin IO1, and the color sensor is connected to pin IO2.}
    \label{fig:robot}
\end{figure}

The program for the color dial uses the reversible primitive \emph{rotate}, used as an example in \cref{sec:implementation}, to rotate the needle of the dial.
By using only reversible output primitives, the program written for this robot automatically becomes reversible. %without any additional input from the programmer.
This means that while debugging the application, the color the needle is pointing towards will always correspond to variable \emph{current} in the program.
Concretely, if the debugger steps back through the program from the end of a loop iteration to line~\ref{line:target}, it will move the needle back to the previous color without having to read a new sensor value.
This makes it easy to test certain state transitions where the needle is pointing at one particular color and now has to move to a different color.

Aside from using time travel debugging which keeps external state in mind, users of our debugger are also able to leverage multiverse debugging capabilities to deal with the non-deterministic nature of this color sensor.
This allows them to easily simulate various sensor values, and explore the different paths the robot can take without needing to use any real, correctly-colored objects.
This example touches on a few common aspects of robotics applications, such as processing non-deterministic input, controlling motors and making decisions based on sensor values.
Using the I/O primitives supported by MIO, various other applications could be build; such as a binary LED counter, a smart curtain, an analogue clock, a maze solving robot, and so on.

\subsection{Discussion and Future Work}\label{sec:Discussion}

Our debugger's formal guarantees are built upon a constrained subset of general I/O operations i.e. deterministically reversible operations.
These operations set the boundaries of where we can guarantee that our debugger is operating correctly.  
The primary limitations in comparison with general I/O operations concern inputs with unbounded or unknown ranges, such as raw network sockets or file systems, as these prevent the exhaustive instrumentation required by our debugger. 
Furthermore, our system excludes outputs that are asynchronous (e.g., an asynchronous web-api) or physically irreversible (e.g., firing a missile), as these violate the core requirements for atomic execution and deterministically reversible actions.

Future research could substantially extend or relax these boundaries. 
For instance, the challenge of unbounded inputs could be addressed through user-defined sampling strategies that allow developers to specify relevant input subsets for exploration. 
These sampling strategies could potentially provide a practical and usable system, even though they don't cover every possibility.
The constraint of physical reversibility could be relaxed by operating in a more controllable environment where reversibility is pushed further. 
The requirement for having external output to be synchronous was mostly introduced to keep the formalization of the system from growing too complex.
Independent asynchronous output actions could thus potentially be supported but would require a substantial formalization effort. 
Finally, a formal language for declaring I/O interdependencies would allow the debugger to reason about more complex, realistic hardware setups, moving beyond our current simple interdependency assumptions.

Despite these potential extensions, some operations remain fundamentally incompatible with the reversible multiverse debugging paradigm. 
The ability to reverse an execution is essential, and this is impossible for operations causing an irreversible state change in an uncontrolled \emph{external} environment. 
Actions such as communicating with a public web service or safety-critical physical action cannot be truly rewound, i.e. firing a missile. 
Such operations, therefore, fall outside the theoretical scope of our approach and should be addressed through a controllable test environment.

Finally, while we have not applied our debugging recipe to other languages we believe our guidelines provide a clear path. 
As an example the same ideas could be applied in a vastly different environment, i.e. C in combination with GDB. 
First, all I/O primitives would have to be rewritten in a dedicated C library, making them easy to identify.
Each I/O primitive would also need to be accompanied by a compensating action.  
GDB could then place breakpoints on each of these I/O primitives to intercept execution at the exact moments required for taking state snapshots.
The snapshotting mechanism could be implemented with the memory read/write commands to save and later restore the program's state.
While more involved, implementing the compensating actions is also feasible. 
GDB can use its memory write capability to directly manipulate the program's memory and stack, effectively executing the "undo" logic for an I/O operation. 
Our formal semantics provide the precise blueprint for this process, dictating exactly when to take snapshots and how to apply these compensating actions when stepping backwards.

\section{Related work}\label{sec:related}

Our work builds directly on WebAssembly~\cite{haas17,rossberg19,rossberg23} and WARDuino~\cite{lauwaerts24a}, as we have discussed in \cref{sec:webassembly} and \cref{sec:warduino}.
In this section, we present an overview of further related work.

\paragraph{Multiverse debuggers}

Multiverse debugging has emerged as a powerful technique to debug non-deterministic program behavior, by allowing multiple execution paths to be explored simultaneously.
It was proposed by \citet{torres19} to debug actor-based programs with a prototype called Voyager~\cite{gurdeep19a}, that worked directly on the operational semantics of the language defined in PLT Redex~\cite{felleisen09}.
Several works have expanded on multiverse debugging; \citet{pasquier22} introduced user-defined reduction rules to shrink the search space for multiverse-wide breakpoint lookup, and \citet{pasquier23} introduced temporal breakpoints that allow users to reason about the future execution of a program using linear temporal logic.
In contrast to MIO, existing multiverse debuggers do not consider I/O operations, or their effects on the external environment. Additionally, previous multiverse debuggers only work on a model of the program execution.

\paragraph{Multiverse Analysis}
The idea of exploring the multiverse of possibilities, is more widely known as multiverse analysis.
Within statistical analysis, it is a method that considers all possible combinations of datasets and analysis simultaneously~\cite{steegen16}.
Within software development, there are several frameworks for exploratory programming~\cite{kery17a} that allow developers to interact with the multiverse of source code versions~\cite{steinert12}.
This way programmers can actively explore the behavior of a program by experimenting with different code.
The approach has led to \emph{programming notebooks}~\cite{perez07,kery18}, and dedicated \emph{explore-first IDEs}~\cite{steinert12,kery17}.
Explore-first IDEs, such as the original by \citet{steinert12}, allow exploration of different source code versions in parallel.
While these editors consider the variations in the program code itself, multiverse debuggers focus on variations of program execution caused by non-deterministic behavior for a single code base.
Combining the two techniques could lead to a powerful development environment, and represents interesting future work.

\paragraph{Exploring Execution Trees}
Many automatic verification and other analysis tools also explore the execution tree of a program, such as software \emph{model checkers}~\cite{godefroid97,jhala09}, \emph{symbolic execution}~\cite{king76,cadar11,baldoni18}, and \emph{concolic execution}~\cite{godefroid05,sen06,marques22}.
These techniques are great at automatically detecting program faults, however, they rely on a precise description of the problem or program specification, often in the form of a formal model.
This is in stark contrast with debuggers, which are tools to help developers find mistakes for which no precise formula exists, and for which the causes are often unknown.
Despite the major differences, static analysis techniques could greatly help improve debuggers by providing the developers with more information.
For multiverse debugging the techniques could help guide developers through large and complicated execution trees.
Additionally, the techniques for handling the state explosion problem~\cite{valmari98,kurshan98,kahlon09} developed for these analysis tools, can help reduce the number of redundant execution paths in multiverse debugging.

\paragraph{Reversible debuggers}
Reversible debugging, also called back-in-time debugging, has existed for more than fifty years~\cite{balzer69}, and has been implemented with various strategies~\cite{engblom12}.
In spite of all the different implementations, few reversible debuggers reverse output effects.
The most common approach is \emph{record-replay debugging}~\cite{agrawal91,feldman88,ronsse99,boothe00,burg13,ocallahan17} that allows offline debugging with checkpoint-based traces.
The recent RR framework~\cite{ocallahan17} is the most advanced record-replay debugger to date.
While replaying it does not reverse I/O operations, in fact, the operations are not performed at all,
instead the external effects are recorded and replayed within the debugger.
One of the earliest works, the Igor debugger~\cite{feldman88}, featured so-called \emph{prestart routines}, which could perform certain actions after stepping back, such as updating the screen with the current frame buffer.
This is one of the first attempts at dealing with external state, however, the solution was purely ad-hoc, and required significant user intervention. %; for instance, supplying the name, mode, and file pointer for each file currently opened during execution.
Additionally, dealing with I/O in a structured way through the prestart routines was still too costly at the time.
There is also no proof of soundness, or any characterization of which prestart routines lead to correct debugging behavior.
Another approach, \emph{omniscient debugging}~\cite{lewis03,pothier07}, records the entire execution of a program, allowing free offline exploration of the entire history, and enabling advanced queries on causal relationships in the execution~\cite{pothier07}.
A third approach is based on \emph{reversible programming languages}~\cite{giachino14,lanese18,lanese18a}.
While not applicable in all scenarios, since it requires a fully reversible language, this approach can enable more advanced features, such as reversing only parts of a concurrent process, while still remaining consistent with the forwards execution~\cite{lanese18a}.
The reversible LISP debugger by \citet{lieberman97} not only redraws the graphical output, but also links graphics with their responsible source code.
Reversible debuggers for the \emph{graphical programming language} Scratch~\cite{maloney10}, namely Blink~\cite{strijbol24} and NuzzleBug~\cite{deiner24}, also redraw the graphical output when stepping back.
However, in all these debuggers, the output effects are internal to the system.
For the Scratch debuggers, the visual output is actually part of the execution model~\cite{maloney10}.

\paragraph{Reversible programming languages}
The concept of reversible computation has a longstanding history in computer science~\cite{zelkowitz73,bennett88,mezzina20}, with the most notable models for reversibility being reversible Turing machines~\cite{axelsen16}, and reversible circuits~\cite{saeedi13}.
Furthermore, the design of reversible languages has evolved into its own field of study~\cite{gluck23}, with languages for most programming paradigms, such as the imperative, and first reversible language, Janus~\cite{lutz86,yokoyama08,lami24}, several functional languages~\cite{yokoyama12,matsuda20}, object-oriented languages~\cite{schultz16,haulund17,hay-schmidt21}, monadic computation~\cite{heunen15}, and languages for concurrent systems~\cite{danos04,schordan16,hoey18}.
Several works have investigated how reversible languages can help reversible debuggers~\cite{chen01,engblom12,lanese18}, however, full computational reversibility is not necessary for back-in-time debugging~\cite{engblom12}.
Moreover, these reversible languages do not consider output effects on the external world, with a few notable exceptions in the space of proprietary languages for industrial robots.

\paragraph{Reverse execution of industrial robots}
While numerous examples can be imagined where actions affecting the environment cannot be easily reversed, there are sufficient scenario's where this is possible, for reverse execution to be widely used in industry.
The reversible language by \citet{schultz20} is particularly interesting.
The work proposes a system for error handling in robotics applications through reverse execution, and identifies two types of reversibility; direct and indirect.
Through our compensating actions, MIO is able to handle both directly and indirectly reversible actions.
\Citet{laursen18} propose a reversible domain-specific language for robotic assembly programs, SCP-RASQ. While we do not focus on a single specific application domain, this work does show how reversible output primitives are possible for advanced robotics applications.
SCP-RASQ uses a similar system of user-defined compensating actions, to reverse indirectly reversible operations.
Using these kinds of languages, we believe that the MIO debugger could be extended to support more complex output primitives, which could control industrial robots.

\paragraph{Reversibility}
The concept of reversibility is well understood on a theoretical level, for both sequential context~\cite{leeman86}, and concurrent systems.
The latter is much more complex, and has lead to two major definitions; causal-consistent reversibility~\cite{danos04,lanese14}, and time reversibility~\cite{weiss75,kelly81}.
Causal-consistent reversibility is the idea that an action can only be reversed after all subsequent dependent actions have been reversed~\cite{lanese14}.
This ensures that all consequences of an action have been undone before reversing, and the system always returns to a past consistent state.
On the other hand, time reversibility only considers the stochastic behavior when time is reversed~\cite{weiss75,kelly81,bernardo23}.
However, it has recently been shown that causal-consistency implies time reversibility~\cite{bernardo23}.
Our debugger works on a single-threaded language, where the non-determinism is introduced by the input operations.
In our work, the undo actions are causally consistent in the single-threaded world.
We believe that we can extend MIO to support concurrent languages, and that the existing literature~\cite{lanese18,giachino14} can help to ensure it stays causally consistent.

\paragraph{Remote Debugging on Microcontrollers}
In remote debugging~\cite{rosenberg96}, a debugger frontend is connected to a remote debugger backend running the program being debugged.
The MIO debugger uses remote debugging to mitigate some limitations of microcontrollers, an approach that has been adopted for many embedded systems~\cite{potsch17,skvar-c24,soderby24}.
These debuggers fall in two categories; \emph{stub} and \emph{on-chip}~\cite{li09}.
A stub is a small piece of software that runs on the microcontroller, instrumenting the software being debugged, this is the approach taken in this work.
On-chip debugging uses additional hardware to debug the embedded device, a common example are JTAG~\cite{shortm} hardware debuggers.
These debuggers can interface with different software, such as the popular OpenOCD debugger~\cite{hogl06}.
Since the debugger front-end needs to communicate with the backend on the microcontroller it exacerbates the probe effect~\cite{gait86} and slows down the debugger.
To address these limitations, a new form of remote debugging was proposed called out-of-place debugging~\cite{marra21:practical,marra18,lauwaerts22}.
This technique moves part of the debugging process to a more powerful machine, which can reduce debugging interference and speedup performance.
The MIO debugger is already sufficiently fast, but a speed-up can likely be achieved by adopting out-of-place debugging.

\paragraph{Environment Modeling}
There are many environment interactions that can influence the possible input values and thereby the possible execution paths of a program.
We have elided these interactions from the formal model and assume that I/O operations are independent, while our prototype does support defining simple \emph{predictable dependencies} between I/O operations.
Modeling the interactions between I/O operations is also hugely important for testing, and \emph{environment modeling} has therefore been widely studied in this area~\cite{blackburn98}.
Environment models are often used for automatic test generation~\cite{dalal99,auguston05}, and have also been applied to real-time embedded software~\cite{iqbal15}.

\paragraph{Formalizing Debuggers}
Previous efforts to define debuggers formally have been incredibly varied in their depth and approach, and have not yet reached a consensus on any standard method.
An early attempt used PowerEpsilon~\cite{zhu91, zhu92} to define a denotational semantic describing the source mapping needed to debug a toy language~\cite{zhu01}.
In 2012, the work by \citet{li12} focussed on automatic debuggers, and defined operational semantics for tracing, and for backwards searching based on those traces.
In 1995, \citet{bernstein95a} defined a debugger in terms of an underlying language semantic for the first time, an approach we adopted in this work as well.
In fact, the approach is followed by a number of later works~\cite{ferrari01, torres17, lauwaerts24a, holter24}, including the work by Torres Lopez et al.~\cite{torres19} in 2019, which defines the correctness of their debugger in terms of the non-interference with the underlying semantic.
The correctness criteria requires each execution observed by either semantic is observed by the other, which is similar to our soundness and completeness theorems rolled into one.
A more recent work presented a novel abstract debugger, that uses static analysis to allow developers to explore abstract program states rather than concrete ones~\cite{holter24}.
The work defines operational semantics for their abstract debugger, and for a concrete debugger.
The soundness of the abstract debugger is defined in terms this concrete debugger, where every debugging session in the concrete world is guaranteed to correspond to a session in the abstract world.
The opposite direction cannot hold since the static analysis relies on an over-approximation, which means there can always be sessions in the abstract world which are impossible in the concrete world.
This is in stark contrast with the soundness theorem in our work, which states that any path in the debugging semantics can be observed in the underlying language semantics.

\section{Conclusion}\label{sec:conclusion}
While existing multiverse debuggers have shown promise in abstract settings, they struggled to adapt to concrete programming languages and I/O operations. 
In this article, we address these limitations by presenting a novel approach that seamlessly integrates multiverse debugging with a full-fledged WebAssembly virtual machine. 
This is the first implementation that enables multiverse debugging for microcontrollers.
Our approach improves current multiverse debuggers by being able to provide multiverse debugging in the face of a set of well-defined I/O primitives. 
We have formalized our approach and give a soundness proof.  
We have implemented our approach and have given various examples showcasing how our approach can deal with a wide range of specialized I/O primitives, ranging from non-deterministic input sensors, to I/O pins and even steering motors. Our sparse snapshotting approach delivers reasonable performance even on a restricted microcontroller platform.
Our initial implementation provides a substantial benefit over existing approaches, but we believe there are further opportunities to relax the constraints on I/O primitives further. For example, our current implementation only supports simple dependencies between I/O actions, but we believe this could be relaxed further by introducing an explicit rule language so that programmers can define more complex dependencies between the I/O actions.

\section*{Acknowledgements}

We thank Francisco Ferreira Ruiz for providing the Open Bot Brain hardware, Jonas Sys for testing the GUI, the members of the VUB DisCo group for fruitful discussions, and all the contributors of the open source WARDuino VM.
\textit{Funding:} Tom Lauwaerts is funded by a project from the Research Foundation Flanders (FWO) [grant number FWOOPR2020008201].

\section*{Data-Availability Statement}

The debugger technique proposed in this article is implemented and available as an open-source project (\href{https://github.com/TOPLLab/MIO}{github.com/TOPLLab/MIO}) developed on top of the open-source WARDuino virtual machine.
The software is also available as research artifact archived and available on Zenodo~\cite{steevens25:mio}.

\bibliographystyle{ACM-Reference-Format}
\bibliography{bibliography}

\newpage
\appendix

\section{Auxiliary Debugger Rules}\label{app:rules}

In this appendix, we present the auxiliary debugger rules for the multiverse debugger for WebAssembly, omitted from \cref{sec:multiverse-debugger} in the main text for brevity.
These are the rules for the step forward operations on primitive calls, and the run variant of the \textsc{step-mock} rule.

\begin{figure}[ht]
        \begin{mathpar}
                \inferrule[(\textsc{step-prim-in})]
       	            { 
                        K_n = \{ s ;v^*; v^*_0 (call \; j) \} \\
                        P(j) = p \\
                        p \in P^{In} \\
                        mocks(j, v^*_0) = \varepsilon \\
                        K_n \hookrightarrow_{i} K_{n+1} \\
                    }
                    { \langle \textsc{pause}, step, mocks, K_n \; | \; S^* \rangle
                      \hookrightarrow_{d,i}
                  \langle \textsc{pause}, \varnothing, mocks, K_{n+1} \; | \; S^* \cdot \{K_{n+1} , r_{nop}\} \rangle }

                \inferrule[(\textsc{step-prim-out})]
       	            { 
                        K_n = \{ s ;v^*; v^*_0 (call \; j) \} \\
                        P(j) = p \\
                        p \in P^{Out} \\
                        p(v^*_0) = \{ \textsf{ret } v, \textsf{cps } r \} \\
                        K_{n+1} =\{ s ;v^*; v \} \\
                    }
                    { \langle \textsc{pause}, step, mocks, K_n \; | \; S^* \rangle
                      \hookrightarrow_{d,i}
                  \langle \textsc{pause}, \varnothing, mocks, K_{n+1} \; | \; S^* \cdot \{K_{n+1} , r\} \rangle }
	\end{mathpar}
        \caption{The \emph{step forwards} rules for input and output primitives in the multiverse debugger for WebAssembly, without input mocking. Addition to \cref{fig:forwards-prim}.}
	\label{fig:forwards-prim-step}
\end{figure}

\begin{figure}[ht]
	\begin{mathpar}
                \inferrule[(\textsc{run-mock})]
       	            { 
                        K_n = \{ s ;v^*; v^*_0 (call \; j) \} \\
                        P(j) = p \\
                        p \in P^{In} \\
                        mock(j, v^*_0) = v \\
                        K'_{n+1} = \{ s';v'^*;v \} \\
                    }
                    { \langle \textsc{play}, \varnothing, mocks, K_n \; | \; S^* \rangle
                      \hookrightarrow_{d,i}
                  \langle \textsc{play}, \varnothing, mocks, K'_{n+1} \; | \; S^* \cdot \{K'_{n+1}, r_{nop}\} \rangle }

	\end{mathpar}
    \caption{The register and unregister rules for input mocking in the MIO multiverse debugger, as well as the \textsc{run-mock} variant. Addition to \cref{fig:mocking} from \cref{sec:mocking}.}
	\label{fig:mocking-additional}
\end{figure}

\section{Proofs and Auxiliary Lemmas}\label{app:proofs}

In this appendix, we present the lemmas and proofs for the multiverse debugger semantics for WebAssembly, omitted from \cref{sec:correctness}.
The first lemma states that the mocking of input values will not introduce states in the multiverse debugger that cannot be observed by the underlying language semantics.
Since the input values accepted by the \textsc{register-mock} rule must be part of the codomain of the primitive, this will always be the case.

\begin{lemma}[Mocking non-interference]\label{lemma:mocking-non-interference}
    Given a debugging state $dbg$ and $dbg \hookrightarrow_{d,i} dbg'$, which uses the \textsc{step-mock} rule, and $K$ in $dbg$, and $K'$ in $dbg'$, it holds that
    \[
        dbg \hookrightarrow_{d,i} dbg' \Rightarrow K \hookrightarrow_{i} K'
    \]
\end{lemma}

\begin{proof}
    Since the \textsc{register-mock} rule only adds a new value to the $mock$ map when the value is in the codomain of the primitive, the value produced by the \textsc{step-mock} can also be chosen by the non-deterministic rule \textsc{input-prim}.
\end{proof}

A second lemma crucial to the soundness of the debugger, states that for any debugging state, there is a path in the underlying language semantics from the start to every snapshot in the snapshot list.

\begin{lemma}[Snapshot soundness]\label{lemma:snapshot-soundness}
    For any debugging state $dbg$ with program state $K_m$, and snapshots $S^*$,  it holds that
    \[
        dbg_{start} \hookrightarrow^*_{d,i} \{rs,msg,mocks,K_m,S^*\} \Rightarrow \forall \{K_n , r\} \in S^* : K_0 \hookrightarrow_i^* K_n
    \]
\end{lemma}

\begin{proof}
    By induction over the snapshots in the steps in $dbg_{start} \hookrightarrow^*_{d,i} \{rs,msg,mocks,K_a,S^*\}$.
    \begin{description}
        \item[Base case] We have $S^* = \{K_0, r_{nop}\}$, and the lemma holds trivially since $K_0 \hookrightarrow_i^* K_0$.

        \item[Induction case] By the induction hypothesis, $dbg_{start} \hookrightarrow^*_{d,i} \{rs',msg',mocks',K_m,S'^*\}$, and $\forall \{K_n , r'\} \in S'^* : K_0 \hookrightarrow_i^* K_n$.
            Now we prove the theorem still holds after: $$\{rs',msg',mocks',K_m,S'^*\} \hookrightarrow_{d,i} \{rs,msg,mocks,K_{a},S^*\}$$
            The possible steps fall in five cases.

    \begin{itemize}
        \item For the rules that do not change the snapshot list, \textsc{run}, \textsc{step-forwards}, \textsc{pause}, \textsc{play}, \textsc{register-mock}, \textsc{unregister-mock}, or \textsc{step-back}, the theorem holds trivially.
        \item For the rules \textsc{run-prim-in} and \textsc{step-prim-in}, $K_a = K_{m+1}$, and the rules extend the snapshot list with $\{K_{m+1},r_{nop}\}$. We know by the assumptions of the rule that $K_m \hookrightarrow_i K_{m+1}$, so the theorem holds.
        \item For the rules \textsc{run-prim-out} and \textsc{step-prim-out} $K_a = K_{m+1}$, and the rules extend the snapshot list with $\{K_{m+1},r\}$. Both rules satisfy the assumptions for the underlying language rule \textsc{output-prim}, and the state $K_{m+1}$ is exactly the same as the state reached by \textsc{output-prim}. So we have $K_m \hookrightarrow_i K_{m+1}$, and the theorem holds.
        \item The rule \textsc{step-mock} adds $\{K_{m+1},r_{nop}\}$ to the snapshot list, $K_a = K_{m+1}$, and we know that $K_m \hookrightarrow_i K_{m+1}$ by \cref{lemma:mocking-non-interference}, so the theorem holds.
        \item The \textsc{step-back-compensate} rule only removes a snapshot from the snapshot list, so by the induction hypothesis, the theorem holds.
    \end{itemize}
    \end{description}
\end{proof}

Now we give the proof for debugger soundness, where the snapshot soundness lemma will be crucial. % followed by the auxiliary lemmas for snapshot soundness (\cref{lemma:snapshot-soundness}), checkpoint existence (\cref{lemma:checkpoint-existence}), and deterministic path (\cref{lemma:deterministic-path}).

\newtheorem*{theorem*}{Theorem}

\begin{theorem*}[Debugger soundness]
    \theoremdebuggersoundness
\end{theorem*}

\begin{proof}
    By induction over the steps in the path $dbg_{start} \hookrightarrow^*_{d,i} dbg$.

    \begin{description}
        \item[Base case] We have $ dbg_{start} =  \langle \textsc{pause}, msg, \lambda z .  \lambda y . \lambda x . \varepsilon, K_0 \; | \; \{ K_0 , r_{nop} \} \rangle$, and the length of the path is $1$.
    The rules \textsc{run}, \textsc{pause}, \textsc{run-prim-in}, \textsc{run-prim-out}, do not apply since the execution state is not \textsc{play}.
    Similarly, the \textsc{step-back} and \textsc{step-back-compensate}, do not apply since the index label for $K$ is zero, and \textsc{step-mock} does not apply because the mocking map is empty.
    The rules \textsc{play}, \textsc{register-mock}, and \textsc{unregister-mock} do not change the state $K_0$, and $K_0 \hookrightarrow^*_i K_0$ holds for length $0$.
    The \textsc{step-forwards} and the \textsc{step-prim-in} rules use the underlying language semantics to step to $K_1$.
    Finally, the requirements for the \textsc{output-prim} in the underlying language semantics are also met by the \textsc{step-prim-out} rule.
    The \textsc{step-prim-out} rule moves the state to $K_{1} = \{s,v^*,v\}$, which is exactly the same state reached by the \textsc{output-prim} rule in the underlying language semantics.
    So the theorem holds for the base case.

     \item[Induction case] We have a debugging state $dbg'$ with WebAssembly state $K'$, we know that $dbg_{start} \hookrightarrow^*_{d,i} dbg'$ holds, and there is a step $dbg' \hookrightarrow_{d,i} dbg$.
    Since $dbg'$ can have any execution state, any message, and any mocking map, we need to consider all possible cases.
    For the rules which do not change the state $K$, the \textsc{play}, \textsc{pause}, \textsc{register-mock}, and \textsc{unregister-mock} rules, and the theorem holds trivially.
    For the \textsc{run}, \textsc{step-forwards}, \textsc{run-prim-in}, \textsc{step-prim-in}, by the induction hypothesis we know that $K_{0} \hookrightarrow^*_i K'$, and the rules take the step $K' \hookrightarrow_i K$, so the theorem holds. % by the same reasoning as in the base case.
    If the mocking map returns a mocked value, the \textsc{step-mock} rule applies, and given the induction hypothesis and \cref{lemma:mocking-non-interference}, the theorem holds.
    However, stepping backwards is more complex.
    In case the final step uses \textsc{step-back}, the rule jumps to a state $K_n$ from the snapshot list.
    By \cref{lemma:snapshot-soundness}, we know that $K_0 \hookrightarrow_i^* K_n$.
    Since in the assumptions of the \textsc{step-back} rule, we know that $K_n \hookrightarrow^{m-n-1}_i K_{m-1}$, the theorem holds.
    The case for the \textsc{step-back-compensate} rule is identical.
    \end{description}
\end{proof}

\begin{theorem*}[Debugger completeness]
    \theoremdebuggercompleteness
\end{theorem*}

\begin{proof}
    For any step $K \hookrightarrow_i K'$ in the path $K_{0} \hookrightarrow^*_i K'$, either we can apply the \textsc{step-forward} or \textsc{step-prim-out} rules to the debugging state $dbg$ with state $K$.
    Or, $K$ is a call to an input primitive, in which case $K \hookrightarrow_i K'$ is non-deterministic.
    However, since we know the return value $v$ in $K'$, we can apply the \textsc{register-mock} rule, after which, the \textsc{step-mock} rule is applicable.
    This rule will move the state to $K'' = {s;v^*;v}$, which is the same as $K'$.
    So the theorem holds for all steps in the path $K_{0} \hookrightarrow^*_i K'$.
\end{proof}

Finally, we give the proof for compensation soundness (\cref{theorem:compensate-soundness}). %, and the needed lemmas.
But first, for completeness, we provide the definition of external effects equivalence for a series of debugging rules and a series of rules in the underlying language semantics.

\begin{definition}[External effects equivalence]\label{def:external-effects}
    Let $t$ be a series of rules in the debugging semantics, and $q$ a series of rules in the underlying language semantics.
    When for each \textsc{step-prim-out} with $p$ in $external(t)$, either the next \textsc{step-back-compensate} in $external(t)$ uses $p_{cps}$, or there is an \textsc{output-prim} with $p$ in $external(q)$, we say that

    $$external(t) \equiv external(q)$$
\end{definition}

\begin{theorem*}[Compensation soundness]
    \theoremcompensatesoundness
\end{theorem*}

\begin{proof}
    The multiverse tree is a connected acyclic graph, where each edge is a step in the underlying language semantics.
    Any debugging session $dbg_{start} \hookrightarrow^*_{d,i} dbg$ can be seen as a walk over the multiverse tree, where edges can be visited more than once, and walking over an edge has a direction.
    By debugger soundness (\cref{theorem:debugger-soundness}), we know that for any debugging session there is a path in the underlying language semantics $K_0 \hookrightarrow^*_{i} K_n$.
    The debugging session constructed in the proof for the debugger completeness (\cref{theorem:debugger-completeness}), shows that for any path in the underlying language semantics, there is a debugging session $P$ that ends in $K_n$, but does not use the \textsc{step-back} or \textsc{step-back-compensate} rules.
    This walk $P$ corresponds to the path from $K_0$ to $K_n$ in the multiverse tree, which only visits each edge once.
    This means that: $$external(P) = external(K_0 \hookrightarrow^*_{i} K_n)$$

Now we can show that any walk over the multiverse tree that ends in state $dbg$ can be reduced to the path $P$ by only removing closed walks.
Take a state $dbg'$ on the path from $dbg_{start}$ to $dbg$.
Say that step $s$ is the first step in the debugging session that ends in the state $dbg'$, and $s'$ is the last step to end in the state $dbg'$.
Then the steps between $s$ and $s'$ must form a closed walk, and we know that no step will come back to $dbg'$.
This holds for each state on the path, and therefore the entire session can be reduced to a path.
Removing a closed walk has no effect on the external world, since each forward visited of an edge will have a corresponding backward visit in the walk.
In other words, the effect on the environment by a closed walk is equivalent to the empty list.
This means that:

$$external(P) = external(dbg_{start} \hookrightarrow^*_{d,i} dbg)$$

Therefore, the theorem holds.
\end{proof}

\end{document}